\documentclass{article}
\usepackage[T1]{fontenc}
\usepackage{graphicx} % Required for inserting images
\usepackage[applemac]{inputenc}
\usepackage{amsfonts}
\usepackage{authblk}
\usepackage{xurl}
\usepackage{setspace}
\usepackage{microtype}
\usepackage{color, colortbl}
\usepackage{dsfont}
\usepackage{amsmath}
\usepackage{amssymb}
\usepackage{bbm}
\usepackage{tikz}
\usepackage{caption}
\usepackage{subcaption}
\usepackage{geometry}
\usetikzlibrary{arrows}
\usepackage{mathtools}
\usepackage{comment}
\usepackage{amsthm}
\usepackage{enumitem}
\usepackage{comment}
\usepackage{float}

\def\cL{{\mathcal {L}}}

\def\cV{{\mathcal{V}}}
\def\cA{{\mathcal{A}}}

\def\bR{{\mathbb {R}}}

\def\cG{{\mathcal {G}}}
\def\cE{{\mathcal {E}}}

\newtheorem{theorem}{Theorem}[section]

\newtheorem{proposition}[theorem]{Proposition}

\newtheorem{assumption}[theorem]{Assumption}

\newtheorem{remark}[theorem]{Remark}

\title{Network-Based Optimal Control of Pollution Growth}

\author[1]{Fausto Gozzi\thanks{fgozzi@luiss.it}}
\author[2]{Marta~Leocata\thanks{mleocata@luiss.it}}
\author[3]{Giulia Pucci\thanks{pucci@kth.se \\ \\ 
All the authors have equally contributed to the paper. \\
\\
Fausto Gozzi and Marta Leocata are supported by the Italian
Ministry of University and Research (MIUR), in the framework of PRIN projects 2017FKHBA8 001 \textit{(The Time-Space Evolution of Economic Activities: Mathematical Models and Empirical Applications)} and  20223PNJ8K \textit{(Impact of the Human Activities on the Environment and
Economic Decision Making in a Heterogeneous Setting:
Mathematical Models and Policy Implications)}. Giulia Pucci is supported by the Swedish Research Council grant (2020-04697).}}

\affil[1]{Department of Economics and Finance, Luiss University, Rome, Italy}

\affil[2]{Department of Economics and Finance, Luiss University, Rome, Italy.}

\affil[3]{Department of Mathematics, KTH Royal Institute of Technology, Stockholm, Sweden}

\date{\today}

\begin{document}

\maketitle

\begin{abstract}
This paper studies a model for the optimal control of pollution diffusion over time and space by a centralized economic agent.
The controls are the investments in two types of production: a less polluting ("green") technology and a more polluting ("brown") one. The goal is to maximize an intertemporal utility function which takes into account the cost of pollution.
The main novelty is the fact that the spatial component has a network structure. Moreover, in such a time-space setting, we analyze the trade-off between the use of green and brown technologies: this is also a novelty in such a setting.
Extending methods from previous works, we can explicitly solve the problem in the case of strictly convex or linear pollution costs.
\end{abstract}


%\tableofcontents

\section{Introduction}
\textcolor{black}{Pollution control is a significant issue that arises in various research areas. In the field of economic and operational research, the problem has been studied from both the regulator and industry perspectives. In this context, a variety of questions have been addressed, ranging from the most fundamental to the more obscure.
As a first example, pollution control can be made by the institutions using environmental regulation or direct control (e.g. to face environmental emergencies), as studied in the paper \cite{Bawa}. There, the author compares the two, determining the proper mix between them.\\
Moreover, the transboundary nature of most pollutants (both air and water pollutants), which goes more in the direction of the regulator perspective, has been extensively studied in the operational research literature, in both applied and theoretical papers. \\
\noindent 
In \cite{kerl2015new}, the authors evaluate fluctuating pollutant formation from source emissions integrated within an electricity production model. In \cite{xu2009local, zhao2013harmonizing}, the authors address the issue of reducing local pollution (in the air and water, respectively) in China. In \cite{stam1992transboundary} the authors propose a multicriteria software package to evaluate various scenarios and trade-offs in Europe.\\
\noindent 
Examples of more theoretical works include \cite{bertinelli2014carbon, DeFrutos1, DeFrutos2, BFFGJME,
ferrari2019strategic, leibowicz2020urban,
el2020transboundary,
de2020non,DeFrutos3,BFFGEJOR,de2022investment,BFFGGEB}.\\
Among the above more    theoretical papers we first recall that %consider \cite{bertinelli2014carbon, ferrari2019strategic, el2020transboundary}. 
in \cite{bertinelli2014carbon, el2020transboundary,de2020non,de2022investment} models in a two-country setting are studied, while in \cite{ferrari2019strategic} a strategic model of pollution control for a firm, representative of the productive sector of a country, is considered. While these three papers study some relevant policy implications, no spatio-temporal dynamics of pollution are investigated there.\\
Geographic heterogeneity in the context of pollution mitigation (which is a key theme in the literature of Operations Research) is considered, for example, in the work of \cite{leibowicz2020urban}, which presents a theoretical spatial framework for sustainable urban land use and transportation planning that explicitly accounts for greenhouse gas-related damages. The model is a static one.\\
Going to dynamic models, to the best of our knowledge, the first works where a spatio-temporal evolution for pollution in a game setting is investigated are \cite{DeFrutos1,DeFrutos2,DeFrutos3}. More specifically, in the above three papers, the authors study a finite player model formulated in continuous time and space, in which the spatial and temporal evolution of pollution is modelled using a parabolic partial differential equation. The players behave strategically and maximise their welfare net of environmental damage caused by the pollutant stock. 
In \cite{DeFrutos1,DeFrutos3}, players choose the level of emission; in \cite{DeFrutos2}, they choose the level of emission and investment in clean technology.
%In \cite{DeFrutos3}, the optimal intraregional distribution of pollutant emissions is studied.
Concerning the methodology, in
\cite{DeFrutos1,DeFrutos2} the authors approximate the problem by discretizing space and performing some numerical analysis, while in \cite{DeFrutos3} they find an explicit solution to the Hamilton-Jacobi-Bellman system and deduce from it a closed loop equilibrium.\\
%the above three papers the impact of spatial heterogeneity is limited to ADD EXPLANATION %strategies has not been explored in the three aforementioned papers.}
%possibly for technical reasons.
%It is possible that the authors' apparent lack of interest in the question, as well as the incompatibility between the technical approach presented and the purpose mentioned above, may have contributed to the situation. 
On the other hand the papers \cite{BFFGJME, BFFGEJOR,BFFGGEB}
employ different tools from infinite dimensional control/game theory and they provide a more general theoretical treatment of related spatiotemporal pollution problems.
%The first of these, \cite{BFFGJME}, does not consider pollution control policies. 
In particular, in \cite{BFFGEJOR}, the authors study a central planner problem in which the negative externality of pollution is modelled as a continuous variable in time and space. 
In this context, the regulator chooses investments and pollution abatement controls, taking into account that production generates pollution and that pollution is transboundary. Notably, the authors derive closed-form solutions to the infinite-dimensional optimal control problem under consideration. This is achieved through a mathematical reformulation of the problem. The approach is not based on dynamic programming or the maximum principle, but on the functional reformulation of the objective function. Some of the ideas of such papers is used here to obtain closed-form solutions.\\
Our work departs from the above contributions introducing two significant innovations:
\emph{modelling the space as a network structure and introducing a less polluting (though more expensive) production method}. \\
Concerning the first novelty, the decision to model the space as a network of interconnected locations rather than as a continuous space is motivated by the fact that the network model is more flexible than the corresponding continuous-time model. This allows us to consider various network structures and provides a better basis for estimation, given that the datasets are clearly discrete. Currently, studies on networks are popular and have been widely applied to the study of economic effects
of heterogeneous interactions between different entities. See, for example, \cite{ballester2006s, elliott2013network, CGLPXY}. See also \cite{nerantzis} for applications in water distribution systems.
Regarding the application of networks to pollution control problems, we observe that \cite{DeFrutos1, DeFrutos2} also treat the discrete space case. However, in these papers, the discrete space is the result of discretising the PDE, rather than being a network with pollution flowing over arcs. In this respect, we must also mention the recent paper \cite{XUEWANGJCP24}, which presents a multi-region differential game is used to investigate the interaction and co-evolution of multiple pollution stocks, with pollution moving across locations through a network. Our work differs from the aforementioned work in terms of both the mathematical approach and the modelling motivations. Indeed, inspired by \cite{BFFGEJOR}, we adopt a different solution method that is not based on HJB equations. This enables us to study three types of geographic discrepancy: discrepancy in productivity, discrepancy in the self-cleaning capacity of nature, and discrepancy in network structure.\\
The second novelty is that the agent also has the possibility, at some cost, to shift part of the traditional production to less polluting processes.
Including two investment options reflects the distinct environmental impacts of pollution-intensive production and low-emission alternatives: coal power, synthetic fertilisers and gasoline vehicles contribute heavily to emissions, whereas renewables, regenerative farming and electric vehicles reduce pollution and environmental damage. However, despite the obvious environmental benefits, these low-pollution processes incur higher costs and generally have lower productivity due to challenges such as energy fluctuations and storage issues. Throughout this work, we will use the terms 'non-renewable' or 'brown' to refer to traditional investments, and 'renewable' or 'green' to refer to the new, lower-emission alternative. When selecting an investment strategy, the agent must balance environmental benefits with economic trade-offs.
  This paper's model allows us to examine such trade-offs and their influence on policies, which seems to be a novelty in such a dynamic context.
 More precisely, we can solve the model in cases where both 'green' and 'brown' technologies coexist, and where only 'brown' technology is available. This provides a kind of negative benchmark, and enables us to quantify the benefits of introducing green production technologies (see Remark \ref{rem:comparison_pollution}).\\
From the mathematical point of view, the novelties in our work require some nontrivial changes in the techniques employed.
For the sake of clarity we then compare our approach
with the ones of the closer papers \cite{BFFGJME,BFFGEJOR,BFFGGEB}:\\
\begin{itemize}
  \item Since the space is now modelled by a network, the state equation becomes a system of ODEs. On one side, this is, in principle, easier than the case of continuous space, as the dimension of the state variable here is finite. On the other hand, this allows us to deal with more general diffusion operators $\mathcal{L}$, which need to be treated differently (see on this the recent paper \cite{CGLPXY}, where the network structure is used in a different context).
  \item The explicit solutions are found by exploiting the linearity of the state equation and of the pollution costs to reduce the problem to a parametric static problem (see Theorem \ref{cor:J_rewritten}). However, the introduction of the less polluting production process $R$ makes it more difficult, with respect to previous papers (see \cite{BFFGEJOR,BFFGGEB}), to find explicit solutions. In particular (see Section \ref{subsec:explicit_solution_ocp}) we can treat the case of generic strictly convex cost $f$, characterizing the solutions also in this case, which is new with respect to the literature.
\end{itemize}
The main results of the paper are: Theorem \ref{cor:J_rewritten} on the reduction to a parametric static problem; Theorem \ref{teo:explicit_1} on the solution of the case when a uniformly convex cost function on renewable is considered; Theorem \ref{teo:explicit_2} on the solution of the case when a linear cost function on renewable is considered; Theorem \ref{teo:P_infty} on the asymptotic behaviour of pollution in the particular case of time-independent coefficients.\\
The content of the paper is as follows: Section \ref{sec:2} presents the general model and examines its well-posedness. In Section \ref{sec:ocp_2_energies}, the optimal control is studied and explicit solutions are found for different cost functions. Moreover, in the same section, a comparison with the case where only one type of investment, the brown one, is admitted is presented. Finally, in Section \ref{sec:simulations}, we consider a quadratic cost for renewable technologies and focus on the analysis of some numerical simulations.
}

\section{A model for Pollution on Networks}
\label{sec:2}
In this section, we provide an overview of the model. Our focus is on a central planning challenge within a spatially organised economy. Within this economy, a single commodity serves multiple roles: it is consumed, utilised as input in both renewable and non-renewable production (invested), and produced at various locations. Additionally, it is important to note that this commodity is not subject to trade between different locations; however, pollution does cross geographical boundaries.\\
In our model, the space variable is described as a network of interconnected geographic locations.  When we refer to a network, we are describing a graph with weights, where the nodes correspond to these locations (e.g., cities, regions, etc.), the edges represent the connections between select locations, and the weights signify the importance of each connection which can vary with time. This network structure enhances the realism of pollution transportation within the model and aligns with the inherent network-like nature of pollutant data, which often exhibits a similar network pattern.\\ 
We model the network of $n\ge2$ geographically distributed locations by a graph $\cG = (\cV, \cE)$ where $\cV$ is a set of vertices that corresponds to locations and $\cE$ is a set of edges connecting vertices. The graph is simple, weighted, and finite. We identify $\cV$ with the set $\{ 1,\ldots,n\}$ of sites, where pollution is accumulated, capital input is invested and output is produced, consumed and locally re-invested and $\cE$ as a subset of $\{ (i,j) \in \{ 1,\ldots,n\}^2 \text{ s.t } i \ne j \}$. We say that two vertices $i, j \in \cV$ are connected, and we write $i 	\sim j$ if there exists an edge connecting them, i.e.,
$    i \sim j  \iff (i,j) \in \cE.
$

We denote by $W=(w_{ij}(t))_{i,j\in \cV}$ the matrix of the graph, with $w_{i,j}(t)$ representing the intensity of the geographical connection from node $j$ to node $i$.
The transboundary nature of pollution is represented by the action of a linear operator $L(\cdot) \colon \bR_+ \to \bR^{n\times n}$ on the nodes of the graph. This captures the idea that pollution may enter or exit any location as a result of interaction between them, and it is defined as  $L(t)=(\ell_{i,j}(t))_{i,j \in \cV}$ with
\[\ell_{i,j}(t)=\begin{cases}
    w_{i,j}(t),\hspace{1.6cm} i\neq j\\
    -\sum\nolimits_{k=1, k\neq i}^n w_{k,i}(t), \quad i=j.
\end{cases}\]
At time $t$ and at any location $i \in \cV$, there is a single individual consuming $C_i(t)$, investing $I_i(t)$ in non-renewable production and $R
_i(t)$ in renewable production. The production coming from the investments is denoted by $Y_i(t)$ and is given by 
\begin{equation}\label{1}
Y_i(t)=a^I_i(t)I_i(t)+a^R_i(t)R_i(t).
\end{equation}
\begin{remark}
     $a^I_i(t)$ and $a^R_i(t)$ represent the productivity or technological levels for non-renewable and renewable investments, at location $i$ in time $t$. They can represent possible technological spillovers across sites, disparities in technological advancement across space % (illustrating obstacles to technological diffusion)
     , and similar dynamics. %In general, non-renewable sources tend to have higher productivity compared to renewable alternatives.
\end{remark}
%We assume $a^I_i(t) >1$,
Inspired by the work \cite{BFFGEJOR}, we assume, for simplification, that capital inputs do not accumulate over time, nor are they exchanged across space. At any location, the output is produced, consumed, and locally invested (no trade across locations), implying the following resource constraints:
\begin{equation}
    C_i(t) + I_i(t) + R_i(t) = Y_i(t).
    \label{2}
\end{equation}
Which, together with \eqref{1}, yields to 
\begin{equation*}
    C_i(t) = (a^I_i(t)- 1) I_i(t) + (a^R_i(t)- 1) R_i(t) 
\end{equation*}
We consider the following control problem with an infinite time horizon on $\cV$. 
 At each time $t \in \bR_+$ and location $i \in \cV$, the planner chooses the control variables: i.e. the investment in traditional or brown production $I_i(t)$ and the investment in green production $R_i(t)$. They contribute to the dynamics of the pollution stock $P_i(t)$.  Specifically, on each node $i \in \cV$, the pollution's dynamics evolve according to the following ODE:
\begin{equation}
     \begin{cases}
         \frac{d}{dt}P_i(t)= \sum\nolimits_{j=1}^n w_{ij}(t)P_j(t)-\sum\nolimits_{j=1}^n w_{ji}(t)P_i(t)-\delta_i(t) P_i(t)+ I_i(t) +\varepsilon_i(t) R_i(t) \\ 
         P_i(0)=p_i \in \bR_+.
     \end{cases}
     \label{eq:P_node}
\end{equation}
Here, the function $\delta_i(t)$ accounts for the effect of nature's self-cleaning mechanisms at node $i \in \cV$, which helps mitigate pollution over time. The function $\varepsilon_i(t)$ represents the pollution intensity associated with the renewable investment at location $i$. Considering values of $\varepsilon_i(t)$ strictly smaller than $1$ is reasonable, as green investments are meant to have a smaller impact on pollution growth compared to brown ones. 
Finally, $p_i \in \bR_+$, represents the initial pollution value for each $i \in \cV$.\\
Consider a social planner, who aims at controlling investment  levels $(I_i,R_i )_{i=1,\dots,n}$  to maximize the following social welfare function %(the set of admissible controls $(I_i,R_i )_{i=1,\dots,n}$ which will be specified later)
 \begin{equation}
\label{funct}     J((p_i,(I_i,R_i))_{i=1,\dots,n}) \coloneqq \int_0^{+\infty} e^{-\rho t} \left( \sum\nolimits_{i=1}^n \left ( \frac{C_i(t)^{1-\gamma}}{1-\gamma}-\omega_i P_i(t)-f_i(R_i(t)) \right) \right ) dt.
 \end{equation}
 Here, $\omega_i$ measures local environmental awareness for each location, $\rho$ is a given discount factor, $\gamma $ denotes the intertemporal substitution in consumption. The function $f_i \colon \bR_+ \to \bR_+ \; \forall i \in \cV$ represents the maintenance and operational costs related to renewable investments. 
   This social welfare function constitutes the social benefit of a community, resulting in a trade-off between different interests, namely technological production and sensitivity to environmental problems. Specifically, the investments $I_i(t)$ and $R_i(t)$  increase utility through consumption but also lead to higher pollution, which incurs cost and decreases utility. 
 \begin{remark}
 Although we recognize that nonrenewable investments also incur costs, in the context of transitioning toward a greener environment, we have chosen to focus exclusively on the running costs coming from renewable energy,  highlighting the critical need to shift towards sustainable solutions to address environmental issues.
 \end{remark}
 \begin{table}[h!]\label{tab:}
 \begin{tabular}{|l|l|}
\hline

Parameter                    & Interpretation                                                                            \\ \hline
$w_{i,j}(t)$                 & intensity of the geographical connection from node $j$ to node $i$                        \\ \hline
$a_i^I(t)$ & technological level for brown investment at the node $i$ and time $t$             \\ \hline
$a_i^R(t)$                   & technological level for green investment at the node $i$ and time $t$                 \\ \hline
$\delta_i(t)$                & decay factor at the node $i$ and time $t$                                                 \\ \hline
$\varepsilon_i(t)$              & pollution intensity associated with the green investment at the node $i$ and time $t$ \\ \hline
$\omega_i$                   & local environmental awareness at the node $i$  \\ \hline
$\rho$ & given discount factor \\ \hline
$\gamma$                     & intertemporal substitution in consumption                                                    \\ \hline

\end{tabular}
\caption{List of parameters involved in the model and their interpretation}
\end{table}
\noindent We use vector notation to describe the pollution stock, $P(t)$ and the consumption level $C(t)$\[P(t) \coloneqq (P_1(t), \ldots, P_n(t)),\quad C(t) \coloneqq (C_1(t), \ldots, C_n(t)),\] and the investments \[I(t)=(I_1(t),\ldots,I_n(t)),\quad  R(t)=(R_1(t),\ldots,R_n(t)).\] 
By defining the operator $
    \cL(t) \coloneqq (L(t)-\delta(t)),$
    and the net emissions function $N(\cdot) \colon \bR_+ \to \bR^n_+$
\begin{equation*}
    N(t) \coloneqq I(t) +\varepsilon(t) R(t), 
\end{equation*} equation \eqref{eq:P_node} can be rewritten as
\begin{equation}\label{eq:P_2}
     \begin{cases}
          \frac{d}{dt}P(t)=\cL(t)P(t)+ N(t) \\
          P(0)=p \in \bR^n_+,
     \end{cases}     
\end{equation} 
where $\delta(t)=\text{diag}(\delta_1(t),\ldots,\delta_n(t))$, 
$\varepsilon(t)=\text{diag}(\varepsilon_1(t),\ldots,\varepsilon_n(t))$.
 By using eq \eqref{1} and \eqref{2}, we can explicitly rewrite the functional in terms of the investments 
  \begin{equation}
     J(p, (I,R)) \coloneqq \int_0^{+\infty} e^{-\rho t} \left ( \left<  \frac{((A^I(t)- \mathbbm{1}) I(t) + (A^R(t)- \mathbbm{1}) R(t) )^{1-\gamma}}{1-\gamma}, \mathbf{1} \right > - \left < \omega, P(t) \right> -\langle f(R(t)),{\bf 1}\rangle\right )dt,
     \label{19}
     \end{equation}
  with $\mathbf{1}$ representing the vector of ones in $\bR^n$, $A^I(t)=\text{diag}(a^I_1(t),\ldots,a^I_n(t))$, $A^R(t)=\text{diag}(a^R_1(t),\ldots,a^R_n(t))$, $\omega=(\omega_1,\dots, \omega_n)$ and $f(R(t))=(f_1(R_1(t)), \ldots,f_n(R_n(t)))$.
The set of admissible controls $\cA(p)$ is defined as 
 \begin{equation*}
 \begin{aligned}
         &\cA(p) \coloneqq \Bigg \{ (I,R) \in PC(\bR_+; \bR^n_+ \times \bR^n_+ ) \colon \int_0^\infty e^{-\rho t}\|N(t)\|dt < \infty, \int_0^\infty e^{-\rho t}\|f(R(t))\|dt < \infty \Bigg \}.
          \end{aligned}
 \end{equation*}
 Here $PC(\bR_+; \bR^n_+ \times \bR^n_+ )$, denotes the set of piecewise continuous functions from $\bR_+$ into  $\bR^n_+ \times \bR^n_+$. \\ We call {\bf (P)} the following state constraint optimal control problem:
 \begin{equation}\label{P}\tag{{\bf P}}
     {\rm Maximize }\quad J(p, (I,R))\quad {\rm over }\quad (I,R)\in \cA(p).
     \end{equation}

% Notice that the problem is a state constraint optimal control problem. However, in the next sections, we will show that we can deal with this technical difficulty, see Theorem \ref{prop:J_rewritten} and Theorem \ref{cor:J_rewritten}.\\

\subsection{On the State Equation and the Well-Posedness of the Objective Function}\label{subsec:well_posedness}
 \begin{assumption}
 \begin{itemize}   \label{assw}
    \item[(i)] For any $i,j \in \cV$ and for any $t \in \bR_+$,  the function $t \mapsto w_{ij}(t)$ is continuous and satisfies $ w_{ij}(t) \ge 0$. There are no self-loops i.e. $w_{ii}(t) = 0 \; \forall t>0$.
    \item[(ii)] For any $i \in \cV$, the functions
    $\delta_i,a_i^I,a_i^R, \varepsilon_i, f_i, $
    are piecewise continuous. They satisfy
    \[ 0 \le  \varepsilon_i(t)<1 \quad \quad a_i^I(t), a_i^R(t)\geq 1, \]   
    and moreover \(\delta_i > 0 \) is bounded.
    %\item[(ii)]     For any $i \in \cV$ and for any $t \ge 0$, the function $t \mapsto \delta_i(t)$ is continuous and satisfies $\delta_i(t) > 0$.\label{assd}
    %\item[(iii)]    For all $i\in \cV$ and $t \ge 0$, the functions $t \mapsto a_i^I(t)$, $ t \mapsto a_i^R(t)$ are continuous and satisfy $a_i^I(t),a_i^R(t)\geq 1$. 
    \item[(iii)] $\rho > 0$, $\gamma \in (0,1)\cup(1,\infty)$, $\omega_i > 0$ $\forall i \in \cV$,
   % \item[(iv)] For $i \in \cV$ and $t \ge 0$, the function $t \mapsto \varepsilon_i(t)$ is continuous and satisfies $\varepsilon_i(t) \ge 0$ and $\varepsilon_i(t) < 1$. The function $t \mapsto \varphi_i(t)$ is continuous and satisfies $\varphi_i(t) \ge 0$. For $i \in \cV$, $\theta_i \in (0,1)$.
%\label{assphi}
%\item[(iv)]      The function $f_i \colon \bR_+ \to \bR_+ \; \forall i \in \cV$ is convex for each $i \in \cV$, \label{assf} 
     \item[(v)]    The linear operator $\cL(t)$ is continuously differentiable and there exists a positive constant $\beta$ such that the time derivative of $\cL(t)$ satisfies $\|\dot{\cL}(t)\| \le \beta$.
    \end{itemize}
\end{assumption}
 Under Assumptions \ref{assw} and by standard results for time-varying linear ODE \cite{rugh1996linear}, we get that for any admissible control strategy, the unique, continuous solution to \eqref{eq:P_2}, is given by
\begin{equation}
    \label{18}
    P(t) = \Phi(t,0)p + \int_0^t \Phi(t,s) N(s)ds, \quad t\ge0
\end{equation}
where the \textit{state transition matrix} $\Phi(t,s)$ solves 
    \begin{equation*}
         \begin{cases}
              \frac{d}{dt}\Phi(t,s)=\cL(t) \Phi(t,s), \qquad t \ge s\\ 
         \Phi(s,s)=I 
         \end{cases}
    \end{equation*}
    \begin{remark} \begin{enumerate}
     \item[(ii)] The state transition matrix $\Phi(t,s) $ is given by the Peano-Baker series and converges absolutely and uniformly for $t,s \in [-T,T]$ where $T>0$ is arbitrary (see Theorem $3.3$ in \cite{rugh1996linear}).
    %\item[(ii)] Observe that at each time $t$, the matrix $L(t)$ is a Metzler matrix, namely $\ell_{ij}(t)\geq 0$ for all $i\neq j$. And since $\delta(t)$ is diagonal, $\cL(t)$ is also a Metzler matrix. This implies that for $t \ge s$, $\Phi(t,s) \in \bR^{n\times n}_+$, \cite{kaczorek2015class}. This implies that, if we choose a zero-investment and abatements path, i.e. if $I(t)=0, R(t)=0, B(t)=0$, for all $t \ge 0 $, then the solution to \eqref{eq:P} is non-negative for every non-negative initial pollution data. \textcolor{red}{ma con questo punto del remark ci facciamo qualcosa??} 
       \item[(ii)]  
       By its definition, for each $t$,  $\zeta=0$ is an eigenvalue for $L(t)$ and the vector $(1,\dots,1)$ is an eigenvector associated with it. All other eigenvalues $\zeta(t)$ of $L(t)$ satisfy $2\min_i \ell_{ii}(t) \leq \text{Re}(\zeta(t))<0$. Given $\delta_i(t) > 0$, the eigenvalues of $\cL(t)$ lie strictly within the left-half complex plane.
       \item[(iii)] Assumption \ref{assw}-$(vi)$ by Theorem $8.7$ in \cite{rugh1996linear}, guarantees exponential stability for the state transition matrix $\Phi(t,s)$ in the time-varying case. Specifically, $\Phi(t,s)$  satisfies $\|\Phi(t,s)\| \le 1$ for any $t \ge s$.
    \end{enumerate}
    \label{rem}
    \end{remark} 
\begin{remark}[Time invariant case] \label{remark:timeinv}For time-independent weight matrix $L(t) \equiv L$ and natural decay $\delta(t) \equiv \delta$, then $\cL(t) \equiv \cL$ and the transition matrix is 
$ \Phi(t,s) = e^{\cL(t - s)}. 
$
In this case, equation \eqref{18} can be rewritten as 
\begin{equation*}
    P(t) = e^{\cL t}p + \int_0^t e^{\cL(t - s)} N(s)ds, \quad t\ge0 
\end{equation*}
where the matrix exponential is defined by the power series 
$
e^{\cL t} = \sum\nolimits_{k=0}^\infty \frac{1}{k!} \cL^k t^k.
$
The operator
\begin{equation*}
    \zeta\mathbbm{1} - \cL \colon \bR^n \to \bR^n
\end{equation*}
is invertible with bounded inverse $(\zeta\mathbbm{1}- \cL)^{-1} \colon \bR^n \to \bR^n$ and the resolvent formula (see Theorem 1.10 in Chapter II of \cite{engel2000one}) holds for every $\zeta>0$: 
\begin{equation}
    (\zeta\mathbbm{1}-\cL)^{-1}h= \int_0^\infty e^{-(\zeta\mathbbm{1}-\cL)t}hdt, \quad \forall h \in \bR^n.
    \label{16}
\end{equation}
Finally, notice that the condition on the eigenvalues given by Remark \ref{rem}$(ii)$, guarantees exponential stability for the state transition matrix $\Phi(t,s)$ in the time-invariant case. Specifically, $\Phi(t,s)$ satisfies $\|\Phi(t,s)\| \le 1$ for any $t \ge s$.
\end{remark} 
%\noindent Notice that for the time-varying system \eqref{eq:P_2}, the condition of placing the eigenvalues of  $\cL(t)$ in the left half plane for every $t >0$ alone is not sufficient to guarantee stability. In this regard, we add assumption \ref{assw}-$(vi)$ that will hold throughout the paper and provide a stability criterion for slowly time-varying linear systems.

\begin{proposition}
    $J(p,(I,R))$ is well defined for all $p \in \bR^n_+$ and $(I,R)\in \cA(p)$, possibly equal to $+\infty$ or $-\infty$ (depending, respectively, on the occurrences $\gamma \in (0,1)$ and $\gamma \in (1,\infty)$).
    \label{welldef}
\end{proposition}

\begin{proof}
    The proof is presented in Appendix \ref{appendix:prop_welldef}
\end{proof}

\section{Solution of the Optimal Control Problem}
\label{sec:ocp_2_energies}
\subsection{Rewriting the Objective Function}\label{subsec:J_rewriting}

The planner aims at solving the optimization problem 
\begin{equation*}
    v(p)\coloneqq \sup_{(I,R)\in\cA(p)}J(p,(I,R)).
\end{equation*}
The function $v$ denotes the value function of the optimization problem, and a triple $(I^*,R^*)$ is said to be an optimal control for the problem starting at $p$ if it satisfies $J(p;(I^*,R^*))=v(p)$.\\

We now define a vector $\alpha$ (which can also be seen as a function of the nodes), which we use to rewrite the objective functional conveniently. Set
\begin{equation}\label{eq:alpha}
    \alpha(s) \coloneqq  \int_s^\infty e^{-\rho (t-s)}\Phi^\intercal(t,s)\omega dt.
\end{equation}
\begin{remark}
Notice that $\Phi^\intercal(t,s)$, the transpose of the state transition matrix, does not always coincide with the state transition matrix associated with $\cL^\intercal$. But this holds if $\cL(t)$ commutes with itself at different times.
\end{remark} 
\textcolor{black}{\begin{proposition}
\label{prop:alpha_bound}
For all \( j \in \cV\) and \(s \ge 0\), \(\alpha_j(s)\) is bounded by  %$$ \frac{  \min_{i \in \cV}  \omega_i}{\rho + \delta_{\max}} \le \alpha_j(s) \le  \frac{  \max_{i \in \cV} \omega_i}{\rho + \delta_{\min}}. $$
 $$ \frac{  \min_{i \in \cV}  \omega_i}{\rho +  \max_{j \in \cV} \sup_{ s \ge 0} \delta_j(s)} \le \alpha_j(s) \le  \frac{  \max_{i \in \cV} \omega_i}{\rho + \min_{j \in \cV} \inf_{ s \ge 0} \delta_j(s)}. $$
\end{proposition}
\begin{proof}
     The proof is presented in Appendix \ref{appendix:alpha_bound}
\end{proof}}

\textcolor{black}{\begin{remark}
Proposition \ref{prop:alpha_bound}, together with Assumptions \ref{subsec:well_posedness}(ii-iii), guarantees that for each $j \in \cV$, $\alpha_j$ remains bounded above and bounded away from zero. 
\end{remark}}
\begin{proposition}\label{prop:J_rewritten}
    We have, for all $p \in \bR^n_+$ and $(I,R)\in \cA(p)$, 
   \begin{equation}
         \begin{aligned}
             J(p, (I,R)) =& -\langle\alpha(0),p\rangle +\int_0^{+\infty} e^{-\rho t} \bigg [ \left<  \frac{((A^I(t)-1)I(t)+(A^R(t)-1)R(t))^{1-\gamma}}{1-\gamma}, \mathbf{1} \right > \\ &- \langle \alpha(t), I(t) +\varepsilon(t) R(t) \rangle -\langle f(R(t)),{\bf 1}\rangle\bigg] dt.
         \end{aligned}
         \label{J}
    \end{equation}
\end{proposition}

\begin{proof}
     The proof is presented in Appendix \ref{appendix:J_rewritten}
\end{proof}

As a consequence of 
\eqref{J}, we get the following useful result.
\begin{theorem}
\label{cor:J_rewritten}
Let $(I^*(t),R^*(t))$ be an admissible strategy, i.e. $(I^*(t),R^*(t))\in \cA(p)$.
Assume moreover that, for a.e. $t \in \bR_+ $, and for each $i \in \cV$., the triplet $(I^*_i(t),R^*_i(t))$ is a maximum point for the function 
$
F_{it}: \mathbb{R}^2_+\to \mathbb{R}
$
where
\begin{equation}\label{eq:F}
  F_{it}(I_i,R_i)=  \frac{((a^I_i(t)-1)I_i+(a^R_i(t)-1)R_i)^{1-\gamma}}{1-\gamma}
    - \alpha_i(t)( I_i +\varepsilon_i(t) R_i)-f_i(R_i),
      \end{equation}
 then $(I^*(t),R^*(t))$ is optimal for the problem {\bf (P)}.
\end{theorem}

\begin{proof}
     The proof is presented in Appendix \ref{appendix:cor_J_rewritten}
\end{proof}

\begin{remark}
    From Theorem \ref{cor:J_rewritten}, it is clear that the optimal control $I_i,R_i$ are interlaced, and they depend on the other nodes only through the parameter $\alpha_i(t)$, which depends on the matrix $\cL$.
\end{remark}
\begin{remark}
    If the value function is finite, from the rewriting of the functional \eqref{funct} presented in \eqref{J} and \eqref{J_node}, we can deduce some monotonic relationships between the value function and the various parameters of the model, in particular
   \begin{itemize}
       \item the value function $v$ is increasing with respect to technological productivities, $a_i^I$, $a_i^R$ for each individual node $i$.
       \item  the value function $v$ is decreasing with respect to pollution intensity $\varepsilon_i$, with respect to pollution awareness $\omega_i$ (since $\alpha$ is a linear in $\omega$) at each location $i$. 
   \end{itemize}
   By assuming quadratic or linear costs, $f_i(R_i)=\lambda_iR_i$ or $f_i(R_i)=\lambda_iR_i^2$ for some parameter $\lambda_i > 0$, we can deduce that the value function is decreasing also with respect to the cost parameter $\lambda_i$.
\end{remark}

\subsection{Explicit Solution of the Problem and Optimal Path}
\label{subsec:explicit_solution_ocp}
To ensure the existence of an optimal solution for {\bf (P)}, we will make the additional assumption.
\textcolor{black}{
\begin{assumption}\label{assumption:renewable}
    \begin{enumerate}
    \item[(i)]  $ \forall \; t \in \bR_+,  \; \forall \; i \in \cV$, $a_i^I(t), a_i^R(t) > 1$.
        \item[(ii)] There exist $C \ge 0, g \in \bR$ such that $ \forall t \in \bR_+$, $ \forall i \in \cV,$ 
        \begin{equation*}
    (a^I_i(t)-1)^{\frac{1-\gamma}{\gamma}}+ (a^R_i(t)-1)^{\frac{1-\gamma}{\gamma}} +  \frac{(a^R_i(t)-1)}{(a^I_i(t)-1)}  \le Ce^{gt}, 
        \end{equation*}      
    \item[(iii)] $\rho > g$.
    \end{enumerate}
    \label{ass3}
\end{assumption}}
\begin{theorem}\label{teo:explicit_1}
    Let Assumption \ref{assumption:renewable} hold. \textcolor{black}{Consider a cost function $f$ satisfying $f(0)=0$,  $f',f'' \ge \epsilon$ for some $\epsilon > 0$, and $\int_0^\infty e^{- \rho t} \|f(R)\|dt < \infty$.}  Then,
    \begin{itemize}
        \item If \begin{equation}\label{cond_1}
\varepsilon_i(t) \le \left(\frac{a_{i}^R(t)-1}{a_{i}^I(t)-1}\right)- \frac{f_i'(0)}{\alpha_i(t)},
\end{equation}
\textcolor{black}{\begin{equation}\label{cond_2}
f'_i\left(\frac{1}{a_{i}^R(t)-1}\left(\frac{a_i^I(t)-1}{\alpha_i(t)}\right)^{1/\gamma}\right)\geq \alpha_i(t)\left(\frac{a_i^R(t)-1}{a_i^I(t)-1}-\varepsilon_i(t)\right),
\end{equation} }
the inner point
    \begin{equation}\label{eq:inner_points}
    \begin{cases}
R_{1,i}=\left[(f'_i)^{-1}\left(\alpha_i(t)\left(\frac{a_{i}^R(t)-1}{a_i^I(t)-1}-\varepsilon_i(t)\right)\right)\right]\\
    I_{1,i}=\left[\left(\frac{\alpha_i(t)}{a_i^I(t)-1}\right)^{-1/\gamma}-(a_{i}^R(t)-1)R_{1,i}\right]\frac{1}{(a_{i}^I(t)-1)},\\
    \end{cases}
    \end{equation}
        belongs to $\cA(p)$ and is optimal starting at each $p$.
            The optimal consumption flow is  for each $i \in \cV$
    \begin{equation*}
        C_{1,i}(t) = \left(\frac{a_i^I(t)-1}{\alpha_i(t)}\right)^{\frac{1}{\gamma}}.
    \end{equation*}
    The optimal emissions flow is for each $i \in \cV$
    \begin{equation}\label{eq:N_optimal_1} 
    N_{1,i}(t) \coloneqq I_{1,i}(t)+\varepsilon R_{1,i}(t).      \end{equation}
 \item If condition \eqref{cond_1} holds and also 
\textcolor{black}{\begin{equation}\label{cond_4}
f'_i\left(\frac{1}{a_{i}^R(t)-1}\left(\frac{a_i^I(t)-1}{\alpha_i(t)}\right)^{1/\gamma}\right)< \alpha_i(t)\left(\frac{a_i^R(t)-1}{a_i^I(t)-1}-\varepsilon_i(t)\right),\\
\end{equation}}
a boundary point with $I_{2,i}=0$ is global maximum. In this case, an explicit formula for the optimal control is not available. However, controls are the unique solution of the non-linear system of equations,
\begin{equation}\label{eq:boundary_points2}
    \begin{cases}
R_{2,i}\,\,s.t.\,\, \left((a^R_i(t)-1)R_{2,i}\right)^{-\gamma}\cdot(a^R_i(t)-1)- \varepsilon_i(t)\alpha_i(t)=f_i'(R_{2,i}) \\ 
I_{2,i}=0,
    \end{cases}
    \end{equation}
    belongs to $\cA(p)$ and is optimal starting at each $p$.
       The optimal consumption flow is, for each $i \in \cV$, 
    \begin{equation*}
        C_{2,i}(t) = \left(\frac{a_i^R(t)-1}{\varepsilon_i(t)\alpha_i(t)+f_i'(R_{2,i})}\right)^{\frac{1}{\gamma}}
    \end{equation*}
    The optimal emissions flow is 
    $
        N_{2,i}(t) \coloneqq \varepsilon_i(t) R_{2,i}(t)
    $
    for each $i \in \cV$.
    \item If \begin{equation}\label{cond_3}
\varepsilon_i(t) >\left(\frac{a_{i}^R(t)-1}{a_{i}^I(t)-1}\right)- \frac{f_i'(0)}{\alpha_i(t)}
\end{equation}
the boundary point
    \begin{equation}\label{eq:boundary_points1}
    \begin{cases}

     I_{3,i}=\left[\left(\frac{\alpha_i(t)}{a_i^I(t)-1}\right)^{-1/\gamma}\right]\frac{1}{(a_{i}^I(t)-1)}\\
     R_{3,i}=0,
    \end{cases}
    \end{equation}
        belongs to $\cA(p)$ and is optimal starting at each $p$.
           The optimal consumption flow is for each $i \in \cV$
    \begin{equation*}
        C_{3,i}(t) = \left(\frac{a_i^I(t)-1}{\alpha_i(t)}\right)^{\frac{1}{\gamma}}.
    \end{equation*}
    The optimal emissions flow is 
    $
        N_{3,i}(t) \coloneqq I_{3,i}(t).   
    $ for each $i \in \cV$.
\end{itemize}
 %    The optimal emissions flow is 
%    \begin{equation*}
 %       N^*(t) \coloneqq \varepsilon R^*(t) - \varphi(t)B^*(t)^\theta   
 %       \label{eq:N_optimal_1},
 %   \end{equation*}
%    and the optimal consumption flow is 
%    \begin{equation*}
 %       C^*(t) =\left(\frac{A^I(t)-\mathbbm{1}}{
%\alpha}\right)^{\frac{1}{\gamma}}.
%    \end{equation*}
\end{theorem}
\begin{proof}
    The proof is presented in Appendix \ref{appendix:teo:explicit_1}
\end{proof}
\textcolor{black}{
\begin{remark} \label{hp_intepretation}  The requirement $a^I, a^R > 1$ in Assumption \ref{assumption:renewable} ensures that neither form of production is a priori excluded from the optimization problem and that all subsequent assumptions and conditions are well defined for any value of $\gamma$. Moreover, Assumption \ref{assumption:renewable}  (ii-iii) guarantees that the optimal solutions lie within the set of admissible controls $\cA(p)$.  
\end{remark}}
\begin{remark}
The result stated in the theorem above admits a clear interpretation. Specifically, we observe that if the pollution rate associated with green investment is sufficiently high, it is optimal to invest exclusively in brown energy, implying \( R \equiv 0 \). \textcolor{black}{Conversely, when the pollution rate of renewable investment is sufficiently low, the optimal strategy depends on the trade-off between the marginal cost of renewable investment and the effect of accumulated pollution on the relative attractiveness of brown versus renewable investment, as captured by conditions \eqref{cond_2} and \eqref{cond_4}. Specifically, if the marginal cost of renewables grows quickly enough, i.e. \eqref{cond_2} holds, then it is optimal to supplement green investment with brown one. On the other hand, if the inequality is reversed, i.e.  \eqref{cond_4} holds, the marginal cost of green investment grows more slowly, and investing exclusively in green production alone is optimal.}
\end{remark}

\begin{theorem}\label{teo:explicit_2}
 If Assumption \ref{assumption:renewable} holds. Consider a linear cost function $f_i(r)=\lambda_i r$, $ \lambda_i > 0$.
     \begin{itemize}
\item 
If  \begin{equation}\label{cond_1_linear}
\varepsilon_i(t)<\left(\frac{a_{i}^R(t)-1}{a_{i}^I(t)-1}\right)- \frac{\lambda_i}{\alpha_i(t)},
\end{equation}
the boundary point
\begin{equation}\begin{cases}\label{eq:boundary1_linear}
    I_{1,i}=0\\
    R_{1,i} = (a^R_i(t)-1)^{\frac{1-\gamma}{\gamma}} \left(\lambda_i +\varepsilon_i(t)\alpha_i(t)\right)^{-\frac{1}{\gamma}} ,
    %R^*_i=\frac{1}{(A^R_i 1)}\left(\frac{\lambda_i+\alpha_i\varepsilon_i}{(a^R_i-1)}\right)^{-\frac1\gamma}+\frac{1}{(A^R_i-1)}\left(\frac{\lambda_i+\alpha_i\varepsilon_i}{(A^R_i-1)\alpha_i\varphi_i\theta}\right)^{\frac{1}{\theta-1}}.
\end{cases}
\end{equation}
    belongs to $\cA(p)$ and is optimal starting at each $p$. 
    The optimal consumption flow is for each $i \in \cV$ 
    \begin{equation*}
        C_{1,i}(t) = \left(\frac{a^R_i(t)-1}{\lambda_i+\varepsilon_i(t)
\alpha_i(t)}\right)^{\frac{1}{\gamma}}. 
    \end{equation*}
The optimal emissions flow is 
 $
        N_{1,i}(t) \coloneqq \varepsilon_i(t)R_{1,i}(t).   
        $ for each $i \in \cV$.
\item If 
 \begin{equation}\label{cond_2_linear}
\varepsilon_i(t)>\left(\frac{a_{i}^R(t)-1}{a_{i}^I(t)-1}\right)- \frac{\lambda_i}{\alpha_i(t)},
\end{equation}
the boundary point
\begin{equation}
\begin{cases}\label{eq:boundary2_linear}
    I_{2,i}=(a^I_i(t)-1)^{\frac{1-\gamma}{\gamma}}\alpha_i(t)^{-\frac1\gamma}\\
    R_{2,i}=0,
\end{cases}
\end{equation}
  belongs to $\cA(p)$ and is optimal starting at each $p$. 
      The optimal consumption flow is for each $i \in \cV$
    \begin{equation*}
        C_{2,i}(t) = \left(\frac{a^I_i(t)-1}{
\alpha_i(t)}\right)^{\frac{1}{\gamma}}.
    \end{equation*}
    The optimal emissions flow is 
       $ N_{2,i}(t) \coloneqq I_{2,i}(t).   
        $ for each $i \in \cV$.
\item If  \begin{equation}\label{cond_3_linear}
\varepsilon_i(t)=\left(\frac{a_{i}^R(t)-1}{a_{i}^I(t)-1}\right)- \frac{\lambda_i}{\alpha_i(t)},
\end{equation}
the inner point
\begin{equation}
\begin{cases}\label{eq:inner_linear}
    (a^I_i(t)-1)I_{3,i}+(a^R_i(t)-1)R_{3,i}=(a^I_i(t)-1)^{\frac1\gamma}\alpha_i(t)^{-\frac{1}{\gamma}}\\
       
       0<R_{3,i}<(a^R_i(t)-1)^{-1}\left( (a^I_i(t)-1)^{\frac1\gamma}\alpha_i(t)^{-\frac{1}{\gamma}}\right),
\end{cases}
\end{equation}
   belongs to $\cA(p)$ and is optimal starting at each $p$.  The optimal consumption flow is, for each $i \in \cV$,
    \begin{equation*}
        C_{3,i}(t) =  \left(\frac{a^I_i(t)-1}{
\alpha_i(t)}\right)^{\frac{1}{\gamma}}.
    \end{equation*}
    The optimal emissions flow is 
 $
        N_{3,i}(t) \coloneqq I_{3,i}(t)+\varepsilon_i(t)R_{3,i}(t).   
$ for each $i \in \cV$.
\end{itemize}
    \label{theo:4.3}
    \end{theorem}
\begin{proof}
        The proof is presented in Appendix \ref{appendix:teo:explicit_2}
    \end{proof}
\begin{remark}
Since the case of linear costs cannot be regarded as a special case of Theorem \ref{theo:4.3}, we analyze it separately. In doing so, we uncover a result that appears natural when linear costs are considered. \textcolor{black}{We essentially distinguish two relevant cases. First, if the effective pollution rate is relatively low, specifically if condition \eqref{cond_1_linear} holds, it is optimal to invest solely in green production, with no investment in browns. Conversely, if 
it is relatively high, as in condition \eqref{cond_2_linear}, it is optimal to invest exclusively in brown production, with zero green investment. The third case, corresponding to the equality in condition \eqref{cond_3_linear}, defines an inner point where both investments coexist; however, this occurs on a set of null measure and is therefore negligible in practical terms.}
\end{remark}

\subsection{Long-Time Behaviour of the Optimal State Trajectory}
\label{subsec:behavior_infty}
We now consider the special case when the coefficients are time-independent. We denote with $(I^*(t),R^*(t))$ optimal consumptions stated in the Theorem \eqref{teo:explicit_1} or in Theorem \eqref{teo:explicit_2}.
\begin{theorem}\label{teo:P_infty}
    Let Assumption \ref{ass3} hold. Assume that the coefficients $a^I_i, a^R_i, \delta_i, \varepsilon_i $ \ are time independent,  $\forall i \in \cV$. Then 
    \begin{equation*}
        \lim_{t\to \infty} \sum\nolimits_{i=1}^n |P^*_i(t) - P^*_{i,\infty}|^2 = 0, 
    \end{equation*}
    where $P^*_i(t)$ is the unique solution of the equation 
    \[\frac{d}{dt} P(t)=\mathcal{L}(t)P(t)+N^*\]
    with $N^*=I^*+\varepsilon R^*$ and $P^*_\infty$ is the unique solution to the matrix equation 
    $
        \cL P + N^* = 0.
    $
    \label{theo:4.4}
\end{theorem}

\begin{proof}
    The proof is presented in Appendix \ref{appendix:P_infty}
\end{proof}

\subsection{The Model with $a_i^R(t) \equiv 1 \; \forall i \in \cV $.}
\label{sec:ocp_1_energy}
We will now address the case in which on every node, the non-renewable productivity factor is strictly greater than one, while the renewable one equals one, namely $a_i^R(t)\equiv 1$ and $a_i^I(t)> 1$ $\forall i\in\cV$. In this setting, any investment in renewable energy sources is economically unfeasible, and thus the entire production relies on a single (traditional) energy source. So, to simplify our setting, we will directly consider the only possible investment $I$. With this simplification, the production takes the form
\begin{equation*}
    Y_i(t)=a^I_i(t)I_i(t),
\end{equation*}
and the resource constraint implies that the consumption is simply given by
\begin{equation*}
    C_i(t) = (a^I_i(t)- 1) I_i(t)
\end{equation*}
The dynamics of $P$ can be written as
 \begin{equation*}
     \begin{cases}
         \frac{d}{dt}P(t)=(L-\delta)P(t)+ I(t)\\ 
         P(0)=p \in \bR_+^n,
     \end{cases}
 \end{equation*}
 and the social welfare to be maximised 
 \begin{equation}
 \label{2.4noR}
     J(p,I) \coloneqq \int_0^{+\infty} e^{-\rho t} \left( \sum\nolimits_{i=1}^n \left ( \frac{C_i(t)^{1-\gamma}}{1-\gamma}-\omega_i P_i(t)\right) \right ) dt.
 \end{equation}
 Finally, the set of admissible controls 
 
 \begin{equation}
 \begin{aligned}
         \cA(p) \coloneqq \Bigg \{  I \in PC(\bR_+; \bR^n_+ \colon \int_0^\infty e^{-\rho t}\left( \sum\nolimits_{i=1}^n |I_i(t)|^2 \right)^{\frac{1}{2}}dt < \infty \Bigg \}.
          \end{aligned}
          \label{set_noR}
 \end{equation} 
To guarantee the existence of a solution, we impose the following assumptions, in addition to those stated in Assumption \ref{subsec:well_posedness}.
\begin{assumption}
    \begin{enumerate}
        \item[(i)] There exist $C \ge 0, \; g \in \bR$ such that 
        \begin{equation*}
    (a^I_i(t)-1)^{\frac{1-\gamma}{\gamma}} \le Ce^{gt}, \qquad \forall t \in \bR_+, \quad \forall i \in \cV,
        \end{equation*}    
    
    \item[(ii)] $\rho > g$.
    \end{enumerate}
    \label{ass2}
    \begin{remark}
      Assumptions \ref{ass2} $(i)-(ii)$ guarantee that the value function is finite.
 Given this, to solve the problem, we use the alternative form \eqref{J} of the objective functional. In such form, we take the control which, for every $t$ maximizes the integrand in \eqref{J}. This is a candidate optimal control. 
    \end{remark}
\end{assumption}
\begin{theorem}
\label{theo:3.3}
    The optimal investment $I_i$, given by 
    \begin{equation*}
       \begin{aligned}
           I_i(t)=\alpha_i(t)^{-\frac{1}{\gamma}}(a^I_i(t)-1))^{\frac{1-\gamma}{\gamma}},
       \end{aligned}
    \end{equation*}
    belongs to $\cA(p)$ in \eqref{set_noR} and is optimal for \eqref{2.4noR} starting at each $p$. The optimal consumption flow is for each $i \in \cV$
    \begin{equation*}
        C_i(t) = \left(\frac{a^I_i(t)-1}{\alpha_i(t)}\right)^{\frac{1}{\gamma}}.   
    \end{equation*}  
        The optimal emissions flow is $
        N_i(t) \coloneqq I_i(t)\label{eq:N_optimal_noR}.
   $ for each $i \in \cV$.
    %The optimal pollution flow is 
    %\textcolor{black}{
    %\begin{equation*}
    %    P^*(t) \coloneqq \Phi(t,0)p + \int_0^t \Phi(t,s)N^*(s)ds, \quad t\ge0.
    %\end{equation*}}
    %The value function is affine in $p$: 
%\begin{equation*}
 %       \begin{aligned}
 %           v(p)= J(p; (I^*,B^*)) =& -\langle\textcolor{black}{\alpha(t)},p\rangle + \int_0^\infty e^{-\rho t} \left ( \sum\nolimits_{i=1}^n \frac{\gamma}{1-\gamma} \left( \frac{a^I_i(t)-1}{\textcolor{black}{\alpha_i(t)}} \right)^{\frac{1-\gamma}{\gamma}}\right) dt \\&- \theta^{\frac{1}{1-\theta}}(1-\theta^{-1})\int_0^\infty e^{-\rho t} \left(  \sum\nolimits_{i=1}^n \textcolor{black}{\alpha_i(t)}\varphi_i^{\frac{1}{1-\textcolor{black}{\theta_i}}} \left( (a^I_i(t) - 1)\right)^{\frac{\textcolor{black}{\theta_i}}{1-\textcolor{black}{\theta_i} }}\right)dt.
 %       \end{aligned}
 %   \end{equation*}
 
\end{theorem}

\begin{proof}
    The proof is presented in Appendix \ref{appendix:3.3}
\end{proof}

%In our model, local pollution reduction efforts are determined by local productivity, taking into account both production and de-pollution activities. This aspect is independent of the transboundary nature of pollution. However, when it comes to 
Notice that in our model, investment depends on the transboundary nature of pollution. The regulator must consider not only the local technological factors but also the potential impact of making investments in a specific location on the neighbouring areas in terms of pollution. It is worth noting that local investments may not necessarily increase with local productivity, denoted as $A^I(t)$. In some cases, higher local productivity might result in lower investments ($I(t)$), which leads to reduced local emissions, albeit at the cost of a slight reduction in production. \\
\begin{remark}[Comparison between pollutions]\label{rem:comparison_pollution}
    We investigate the impact of introducing a new source of energy into the economic system. In particular, we analyze how the adoption of a cleaner energy source affects pollution levels. Our study compares the optimal emission flow in two scenarios:
\begin{enumerate}
    \item When only a single energy source is available, as described in equation \eqref{eq:N_optimal_noR}.
    \item When two energy sources are considered, and the cost of using the new energy source follows a quadratic function, as described in equation \eqref{eq:N_optimal_1}.
\end{enumerate}
Since we assume that \( f \) is strictly increasing, it follows that
\[
N_{1,i} = N_{i} - (f'_i)^{-1} \left(\alpha_i(t) \left(\frac{a_i^R(t) - 1}{a_i^I(t) - 1} - \varepsilon_i(t) \right)\right) \left(\frac{a_i^R(t) - 1}{a_i^I(t) - 1} - \varepsilon_i(t) \right) < N_{i}.
\]
This result implies that, in each location, the optimal emission flow when two energy sources are available is lower than the optimal emission flow when only a single, more polluting energy source is used.
\end{remark}

\section{Some Numerical Investigations when the Renewable Cost Function is Quadratic}
\label{sec:simulations}
In this section, we present the results of a series of quantitative exercises where the cost of renewable technology is quadratic, i.e. $f_i(R_i)=\lambda_iR_i^2,$
and all the relevant parameters of the model are kept constant.
 A straightforward consequence of the static choice of parameters is that the profile of optimal investments is also constant in time, and in particular, they coincide with their long-time distribution. The same argument holds for production, consumption and emission, but not for pollution. Indeed, in the plots presented in the section, we will just present its long-time distribution.\\
The large number of parameters in the model gives us great flexibility in terms of potential scenarios to study. We have therefore selected some interesting ones focusing on how optimal investment strategies are affected by the use (or non-use) of renewable investments and how these choices change according to heterogeneities on the nodes. In particular, our analysis will focus on the following points.
\begin{itemize}
    \item The role of Renewable Investment: We aim to analyze how optimal investment choices change by including renewable investment in the model. %We also investigate how the costs related to the renewable investment affect the optimal strategies. \textcolor{black}{vedere se lo facciamo}
\item Geographic or Economic Heterogeneity: By introducing parameter heterogeneity among nodes, we study how geographic and economic differences influence investment decisions. We explore the effects of spatial factors (e.g. modelling rural or urban settings) by varying natural decay across nodes and modifying the diffusion matrix $L$ to represent different types of pollution flow.  As for the economic heterogeneities, we consider different values of productivity factors (for both green and brown production).
\item  Combined effects of Heterogeneity: We consider scenarios that combine geographic and economic disparities. We will study how the results of the previous analysis are enhanced or mitigated when combined. This allows us to investigate, for example, how economically but not geographically advantaged network nodes, such as densely populated urban centres, respond differently in terms of optimal investments compared to less productive but geographically advantaged areas. 
\end{itemize}

\subsection{Calibration of the Parameters}
\label{calibration}
We start by focusing on the characterization and calibration of the parameters based on existing literature. Given the similarity of the two models, we will refer mainly to \cite{BFFGEJOR} to assign numerical values. \\ 
For the network structure, we consider $n=21$ nodes arranged on a circular network. We model the flow of pollution on the network through the matrix $L$. The elements of $L=(\ell_{ij})$ control how the pollution is dispersed from one node to the other, reflecting geographic factors such as proximity and connectivity between locations. Each entry $\ell_{ij}, i \ne j$ reflects how easily the pollution generated at node $j$ can spread to node $i$, and the diagonal term $\ell_{jj} = - \sum\nolimits_{i \ne j} \ell_{ij}$ keeps track of all the pollution leaving node $j$. Notice that the matrix is not required to be symmetrical, this allows us to consider asymmetric pollution flows between regions, capturing phenomena such as the influence of directional factors, such as prevailing winds or currents, which may make pollution disperse more strongly in one direction than another with no need for additional advection terms.  \\
We first consider a \textit{nearest-neighbor Laplacian} with periodic boundary conditions, where the matrix $L^1=(\ell^1_{ij})$ is given by:
\[
\ell^1_{ij}=\begin{cases}
    \frac{1}{2}, \quad j=i+1\,\,\text{or}\,\, j=i-1\\
    -1,\quad i=j \\
    \frac{1}{2},\quad i=1,\, j=n \,\,\text{or}\,\,i=n,\,\, j=1 
\end{cases}
\]
 Here, pollution spreads from a node to its immediate neighbour, meaning the diffusion is relatively local and slow.\\ % This choice mimics natural dispersion patterns in the environment and represents a good approximation of a situation in which pollution is expected to spread through local interaction (e.g. wind patterns in urban streets or waterways in canals). \\ 
 Afterwards, we also test a \textit{distance based laplacian} on all nodes where the matrix $L^2=(\ell^2
 _{ij})$ is given by 
\[
\ell^2_{ij} = \begin{cases} 
\frac{1}{\min(|i - j|, n - |i - j|)} & \text{if } i \neq j.
 \\[10pt]-\sum\nolimits_{k \neq i} \frac{1}{\min(|i - k|, n - |i - k|)} & \text{if } i = j,
\end{cases}
\]
This representation defines each element $\ell^1_{ij}$ as the reciprocal of the circular distance between nodes $i$ and $j$.  
The idea is that the pollution on each node can spread to all the others with a weight that diminishes with distance. Compared to the matrix $L^1$, we have global diffusion, even though the intensity reduces as the distance grows. \\%This suits scenarios in which pollution disperses over wide areas, such as open plains with wind or bodies of water.\\  
To incorporate asymmetries in the spread matrix, we add distortion in the diffusion of pollution, by the use of a "wind" with strength, e.g. $wind=0.4$, acting only on a specific set of nodes $\tilde{\mathcal{J}} \subseteq \{1 \ldots, n\}$. This adjustment increases the dispersal rate in the direction of the wind and lowers it. We define the  $L^3=(\ell^3_{ij})$ such that the off-diagonal elements are given by: 
\[
\ell^3_{ij}=\begin{cases}
    \frac{1}{2} \pm wind, \quad j=i \pm 1, i \in \tilde{\mathcal{J}}\\ 
     \frac{1}{2},  \quad j=i \pm 1, i \notin \tilde{\mathcal{J}}.  
\end{cases}
\]
The diagonal elements are set to ensure that the sum of each column is zero, maintaining the balance between pollution entries and exits from each node. Boundary conditions are applied as needed but are not explicitly stated here.\\ 
Another advantage of allowing asymmetries in the $L$ matrix is that it can represent geographic areas in which pollution has difficulty entering or exit. For example, we can define two set of nodes  $\tilde{\mathcal{I}},\tilde{\mathcal{J}}  \subseteq \{1 \ldots, n\}$ and construct $L^4 =(\ell^4_{ij})$ such that 
\[
\ell^4_{ij}=\begin{cases}
     \zeta \ell_{i,j}, \quad j \in \tilde{\mathcal{J}}, i \in \tilde{\mathcal{I}}, \\
    \ell_{ij},\quad \text{otherwise.}
\end{cases}
\]
We will set the $\ell_{i,j}$, $i \ne j$ as in $L^1$, with the diagonal values $\ell_{i,i,}$ such that the sum of each column is zero. Depending on the value of $\zeta \ge 0$, pollution flows are restricted or enhanced.\\ 
For the calibration of the other parameters, we will let $\delta_i$ live in a range from $30\%$ to $50\%$ per year.  We will make this parameter vary across the nodes to reflect geographical disparities between urban and rural areas, as well as other factors that may create geographically advantaged or disadvantaged regions. \\ 
%While human activity is a major source of pollutant emissions, natural processes contribute significantly to reducing their impact on the environment: chemical reactions help break down pollutants, and atmospheric dispersion gradually dilutes concentrations. Vegetation plays an important role too, as it absorbs carbon dioxide, while wetlands help filter out other contaminants. These decay processes, combined with geographic and climatic conditions, establish a natural decay rate for various pollutants, which can be calibrated to reflect the many local environmental factors. Referring to \cite{BFFGEJOR}, we will let $\delta = (\delta_i)$ live in a range from $30\%$ to $50\%$ per year. We will allow this parameter to vary across the nodes to reflect geographical disparities between urban and rural areas, as well as other factors that may create geographically advantaged or disadvantaged regions. \\  
In \cite{BFFGEJOR}, the productivity factor $a^I$ is calibrated to achieve an investment-to-gross domestic product ratio within $15-40\%$, which aligns with typical economic ranges. Here, because of the presence of two forms of investment contributing to the GDP, such reasoning is not as straightforward.  Furthermore, the productivity of renewable investments is affected by several factors, including the kind of process taken into consideration, the geographic location, the type of energy employed and so on. Here we set $a^I$ as in \cite{BFFGEJOR}, thus in the range $[2.5,6.5]$ and $a^R\in [1,2.75]$, whereby setting $a^R=1$ we will study the case in which no sustainable investment is taken into account. 
\\
The environmental awareness $\omega$, representing the agent's awareness of pollution levels and their impact on health, environment and overall quality of life, is set to a baseline value of $\omega=1$ and can be interpreted as either a social factor or an implicit tax on pollution.%, as it is natural to expect higher values in heavily urbanized areas due to more visible pollution effects and media coverage. 
\\
The discount factor is taken equal to $3\%$. %, the depollution efficiency parameters $\varphi$ and $\theta$ are taken constant among the nodes and chosen to be equal to $\varphi=0.3$ or $\varphi = 0.11$ and $\theta=0.2$.Note that $\varphi$ and $\theta$ must be chosen carefully, as selecting them too high relative to the other parameters may violate condition \eqref{assumption:renewable}, thus affecting the positivity of pollution as well. 
Initial condition of pollution, $p$ is chosen constant in space, $p_i =1$ $\forall i\in\cV$. Finally, we fix the inverse of the inter-temporal elasticity of substitution, which measures the regulator's aversion to inequality in consumption,  to be equal to $\gamma=0.5$. This choice implies a relatively low aversion to consumption inequality and reflects a social planner who is willing to tolerate some fluctuations in consumption to prioritize other objectives, such as optimizing investments and reducing pollution.   \\ 
It is quite challenging to choose numerical values for $\varepsilon_i$ and $\lambda_i$ as, once again, empirical data on renewable impact per unit investment varies significantly across investment types, geographical regions, employed energy and so on. For $\varepsilon$, which measures the impact of renewable production on emissions, we choose $\varepsilon_i=0.1$, to reflect that renewable production contributes less to pollution compared to traditional sources, whose impact is normalized, and so it is $1$. For $\lambda_i$, which represents the costs of (carrying, maintaining, installing, etc.) renewable investments, we set a value of $\lambda_i =1$.  %For $\lambda$, which represents the costs of (carrying, maintaining, install, etc.) the renewable investments, we experiment with values in the range  $[0.01,10]$ to capture low and high-cost scenarios. %This wide range allows us to examine the impact of different financial barriers on optimal investment strategies and how these may influence the model's drive toward renewable energy adoption.
\\
The following table summarizes the parameter values that are kept constant across all simulations, while the particular choices for the remaining parameters are provided in the captions of the corresponding figures.
 \begin{table}[h!]\label{tab:}
 \begin{center}
\begin{tabular}{|l|l|}
\hline
Parameter                    & Value                                                                            \\ \hline

$n$                   & $21$ \\ \hline
$\omega_i$                   & $1$ \\ \hline

%$\theta_i$                     &  $0.2$                                                           \\ \hline

$\varepsilon_i$              & $0.1$ \\ \hline
$\lambda_i$ & $1$ \\ \hline
$\rho$ & $0.03$ \\ \hline
$\gamma$                     & $0.5$  \\                         \hline
\end{tabular}
\caption{ Fixed Parameter Values for All Simulations}
 \end{center}
\end{table}
\subsection{Numerical Experiments}
In the following figures, we illustrate how optimal investments, long-term pollution levels, production, consumption and emissions vary by location. Non-renewable investments are always shown in black, renewable investments in green, and pollution reduction investments in red. When necessary, we also include plots of the model parameters that vary between nodes, such as natural decay and technological productivity. \\
We begin our analysis by investigating how the inclusion of renewable investment affects the optimal investment choices and, consequently, all the other quantities within the model. Figure \ref{fig:1} illustrates two scenarios: the straight lines represent the case where renewable productivity $a^I_i$ is set equal to $1$ across all nodes, which corresponds to a scenario with no renewable investments. In contrast, the dashed lines represent the case where renewable productivity exceeds $1$ on some nodes, leading to positive renewable investments. The particular choice of parameters is summarized in the caption. In the first subplot, the non-renewable productivity is shown in black, while the renewable productivity is represented in green. The straight line corresponds to $a^R_i=1$ on all nodes, while the dashed line addresses the case where $a^R_i \ge 1$ on the nodes, with a peak of $a^R_i = 2.75$ at the central node and a gradual decrease towards the outer nodes. In each of the subplots, we report the behaviour of some of the model variables on the nodes. We observe that the optimal renewable investment $R$  follows the spatial distribution of $a^R$, showing a peak at the central node. When renewable investments are introduced, non-renewable investment $I$ decreases, with a peak of $-25\%$ at the central node. Emissions and long-term pollution levels decrease in regions where renewable productivity exceeds 1. Specifically, long-term pollution is reduced by up to $20\%$ at the central node. Production increases as a productivity source is added to the model, while the optimal consumption level remains unchanged. 
\begin{figure}[H]
    \centering
\includegraphics[width=0.7\linewidth]{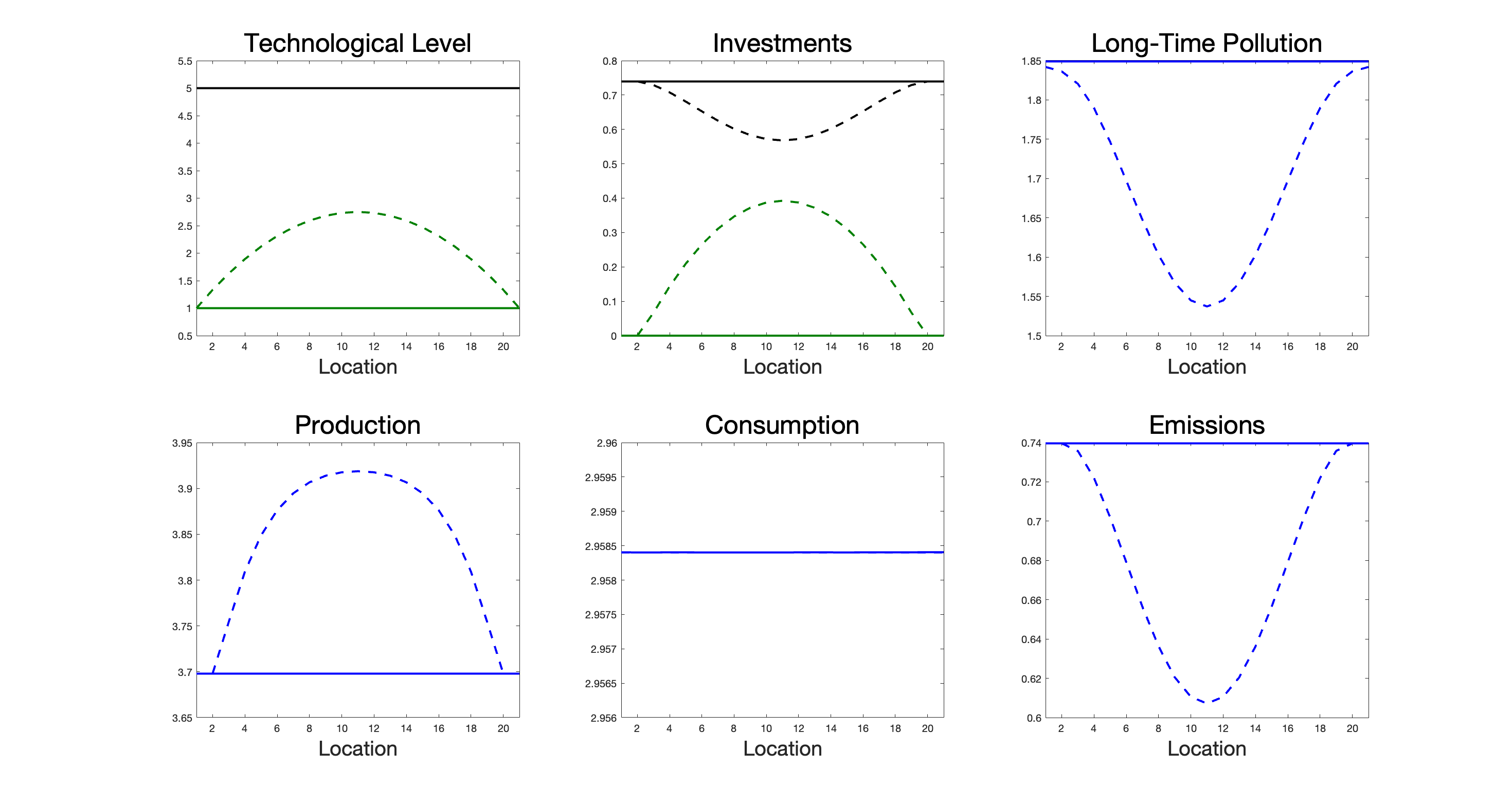}
    \caption{Impact of renewable investments on optimal investment choices, long-term pollution levels, production, consumption and emissions across nodes. The parameter values specific to this figure are:  $L=L^1$, %n=21, \delta=0.4, \rho = 0.03, \omega=1, \varepsilon=0.1, \lambda=1, \theta=0.2, \gamma=0.5$
$\delta_i=0.4$,  $a_i^I = 5$. In the straight-line scenario $a_i^R = 1$, in the dashed-lines scenario $a_i^R = 2.75$ at the core and $a_i^R = 1$ at the periphery,  $\forall \; i \in \cV$.}
    \label{fig:1}
\end{figure}
 Since all parameters are time-independent, by Remark \ref{remark:timeinv}, $\alpha$ is given by 
\begin{equation}
    \alpha = \int_0^\infty e^{-(\rho\mathbbm{1}-\cL^\intercal)t}\omega dt,
\end{equation}
i.e. it is the unique solution to the linear system
$
    (\rho\mathbbm{1} - \cL^\intercal)\alpha=\omega.$
\textcolor{black}{Moreover, for each $i \in \cV$, the bounds established in Proposition \ref{prop:alpha_bound} take the form:
 \begin{equation*}
        \min_{j  \in \cV} \frac{\omega_j}{\rho + \delta_j} \le \alpha_i \le \max_{j \in \cV} \frac{\omega_j}{\rho + \delta_j}
    \end{equation*}}
Notice that the $i-$th component of $\alpha$ reflects the cumulative effect of environmental awareness $\omega$ at location $i$, adjusted for how pollution spreads over time. Specifically, it reflects the combined effects of pollution dispersion across locations (through $L$), of natural decay $\delta$ and time discounting $\rho$. From the bound, it follows that, if $\omega_i$ and $\delta_i$ are constant across all nodes $ i \in \cV$, then $\alpha$ remains constant as well.  Consequently, the choice of network structure governing the pollution dispersion in space does not affect the disutility. However, allowing natural decay and/or pollution awareness to vary across nodes means that different network structures will result in different disutility values at the nodes, and consequently, this will impact the optimal strategy. \\
In Figure \ref{fig:2}, we compare two different choices for the matrix $L$, as introduced in Section \ref{calibration}. The straight lines represent the case $L=L^1$, while the dashed lines $L=L^2$. We keep all model parameters constant across nodes, except for the natural decay $\delta$, which has a minimum at the central node and maximums at the extremes, as shown in the first subplot. We observe that non-renewable investment follows the shape of the natural decay: it is lower in regions where $\delta$ is low and vice versa. In contrast, renewable investment exhibits the opposite behaviour, being higher in areas where self-cleaning capacity is weaker. Long-term pollution, emissions, production and consumption also follow the shape of $\delta$. We note that all the quantities are affected by the choice of $L$: when using the distance-based Laplacian $L^2$, differences across nodes are reduced, leading to smoother overall distributions. Specifically, the long-term pollution appears to be almost homogeneous across locations. While $L^1$ spreads pollution only to the adjacent nodes, $L^2$ allows pollution to spread broadly across the network, resulting in mitigated sharp variations and a more balanced distribution for investments, production and all the other economic factors. 
\begin{figure}[H]
    \centering
\includegraphics[width=0.7\linewidth]{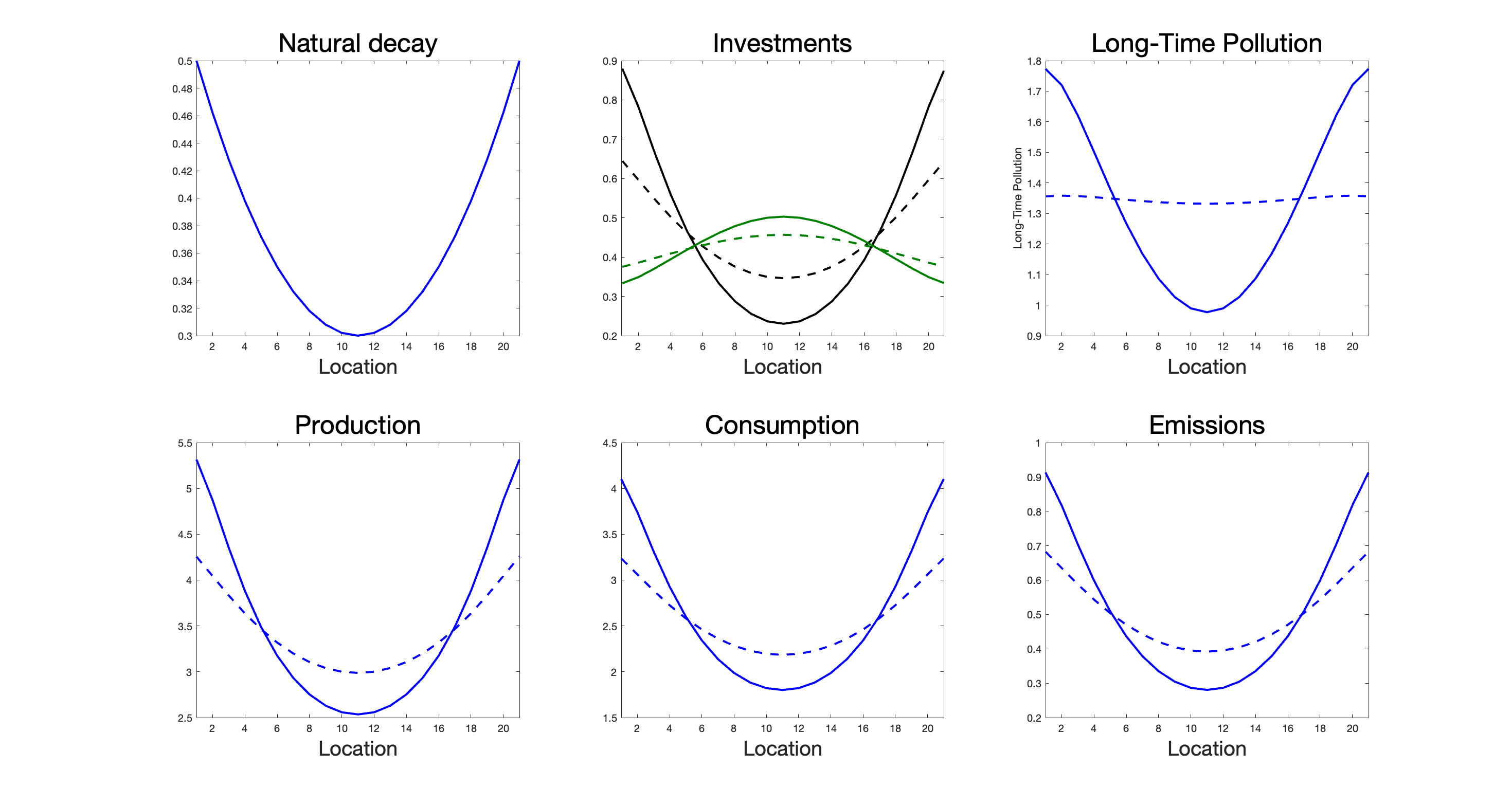}
    \caption{Comparison of optimal investment choices, long-term pollution levels, production, consumption and emissions across nodes for different choices of the operator $L$. The parameter values specific to this figure are:    %$n=21, \rho = 0.03, \omega=1, \varepsilon=0.1, \lambda=1, \theta=0.2, \gamma=4$,
     $a_i^I = 5$, $a_i^R =2.75$, $\forall\; i \in \cV$. The parameter $\delta$ varies across nodes, namely $\delta_i=0.3$ at the core and $\delta_i=0.5$ at the periphery. The straight lines represent the case $L=L^1$ and the dashed lines the case $L=L^2$.}
    \label{fig:2}
\end{figure}
In Figure \ref{fig:3}, we further investigate geographic heterogeneities by comparing the effect of $L=L^1$ and $L=L^3$, respectively, in straight and dashed lines. As presented in Section \ref{calibration}, in $L^3$ we impose a "wind" effect that influences only a subset of the nodes, namely $\mathcal{J}=\{ 14, \ldots, 19\}$. This accentuates dispersion towards higher-numbered nodes while reducing dispersion toward lower-numbered nodes. 
As a result, we observe a bottleneck effect at node $13$, where pollution accumulates. Additionally, in the wind-affected regions, the behaviour of all economic quantities is pronounced compared to the case $L= L^1$: renewable investment is higher,  while nonrenewable investment, production, consumption, and emissions are all lower. %This is due to the fact that pollution is spread more aggressively in the wind-affected area, altering investment and production decisions accordingly.
\begin{figure}[H]
    \centering
\includegraphics[width=0.7\linewidth]{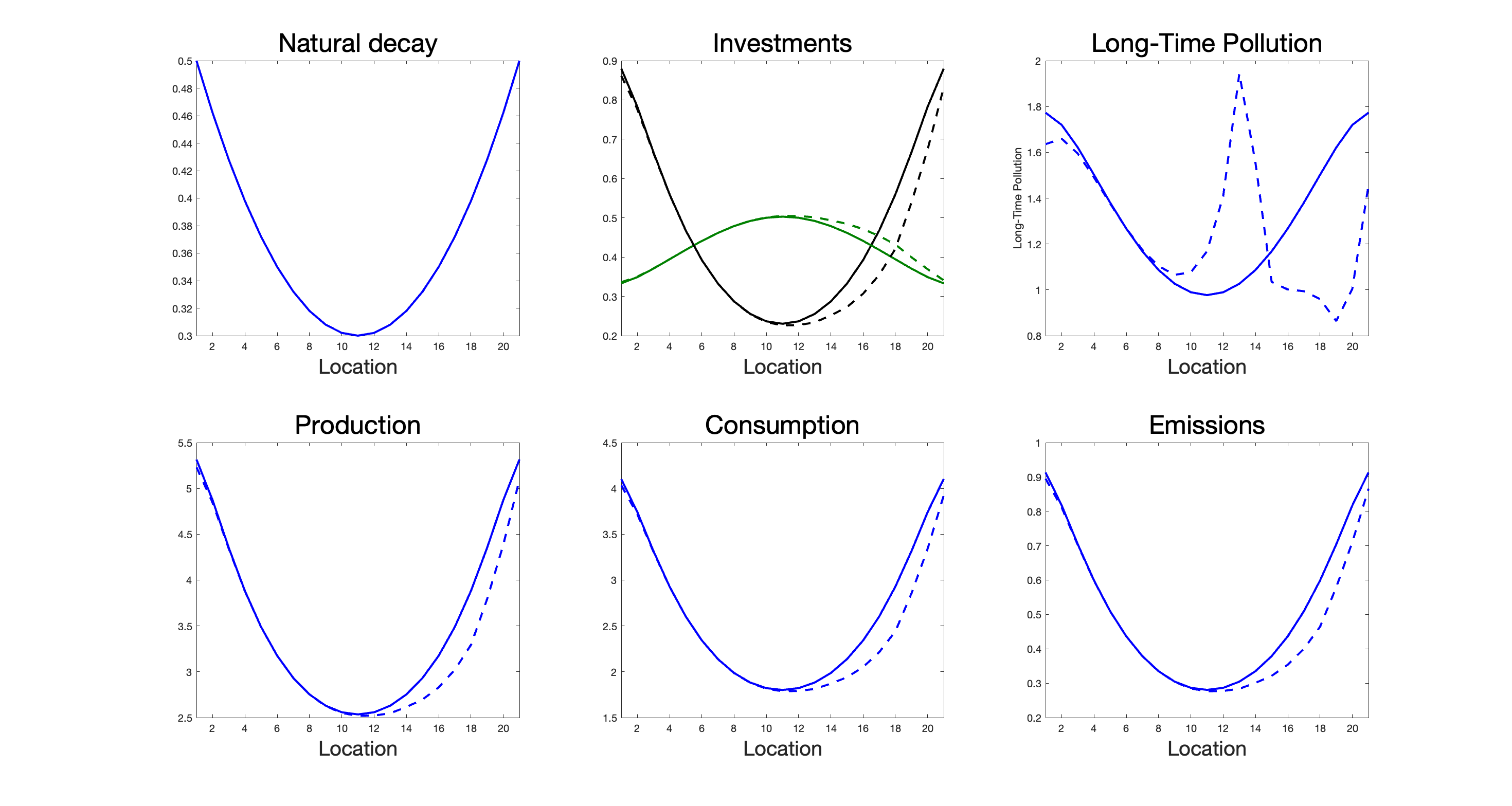}
    \caption{Comparison of optimal investment choices, long-term pollution levels, production, consumption and emissions across nodes for different choices of the operator $L$. The parameter values specific to this figure are:   %$n=21, \rho = 0.03, \omega=1, \varepsilon=0.1, \lambda=1, \theta=0.2, \gamma=4$,$ 
    $a_i^I = 5$, $a_i^R =2.75, \text{wind}=0.4$, $\forall \; i \in \cV$. The parameter $\delta$ varies across nodes, namely $\delta_i=0.3$ at the core and $\delta_i=0.5$ at the periphery. The straight lines represent the case $L=L^1$ and the dashed lines the case $L=L^3$.}
    \label{fig:3}
\end{figure}
In Figure \ref{fig:4}, we compare the choices $L=L^1$  and $L=L^4$. For $L^4$ we set $\zeta=0$, $\tilde{\mathcal{I}}= \{ 8,14\}$, $\tilde{\mathcal{J}}= \{9,13\}$, this results in a matrix where $\ell_{8,9} = \ell_{14,13} = 0$, i.e. we create a zone between nodes $9$ and $13$ where pollution cannot exits. Compared to the case $L=L^1$, this results in higher pollution levels at nodes $9$ and $13$ due to accumulation, and lower levels at $8$ and $14$. However, the changes in pollution are localized to the affected area, leading to only slight adjustments in other quantities without causing substantial re-optimization of investments or production. 
\begin{figure}[!htpb]
    \centering
\includegraphics[width=0.7\linewidth]{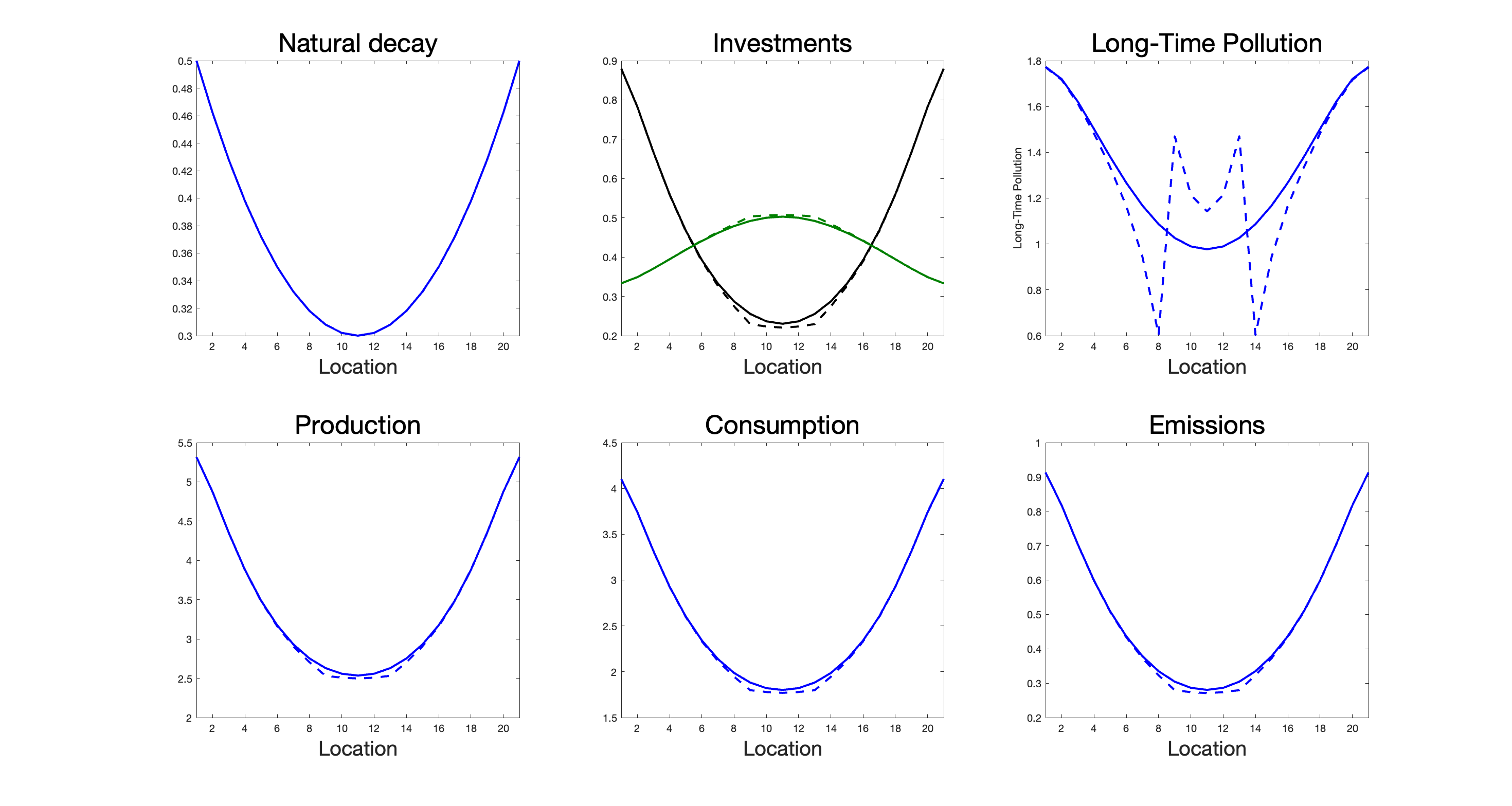}
    \caption{Comparison of optimal investment choices, long-term pollution levels, production, consumption and emissions across nodes for different choices of the operator $L$. The model parameters are:   $a_i^I = 5$, $a_i^R =2.75$, $\forall \; i \in \cV$. The parameter $\delta$ varies across nodes, namely $\delta_i=0.3$ at the core and $\delta_i=0.5$ at the periphery. The straight lines represent the case $L=L^1$ and the dashed lines the case $L=L^4$.}
    \label{fig:4}
\end{figure}
We conclude this analysis by introducing an additional form of heterogeneity across the nodes: spatial discrepancy in input productivity. To capture this, we vary the values of $a^I$ and $a^R$ across the nodes, allowing for different levels of productivity in each location. In Figure \ref{fig:5}, we compare the case of geographical heterogeneity already addressed in Figure \ref{fig:2} with the case where non-renewable productivity is unevenly distributed across locations, concentrating in the central nodes to create an economically advantaged zone. It is interesting to notice that the geographic heterogeneity alone (straight lines) leads to a natural sorting where renewable investment dominates in geographically disadvantaged zones and non-renewable flourishes where pollution dissipates faster. Adding economic heterogeneity (dashed lines) completely overrides this effect: non-renewable productivity is now higher in the core, and consequently, it becomes more economically attractive to invest there, despite the higher pollution. %This leads to more investment in depollution efforts (in red) to counteract the pollution buildup in the core. 
\begin{figure}[H]
    \centering
\includegraphics[width=0.7\linewidth]{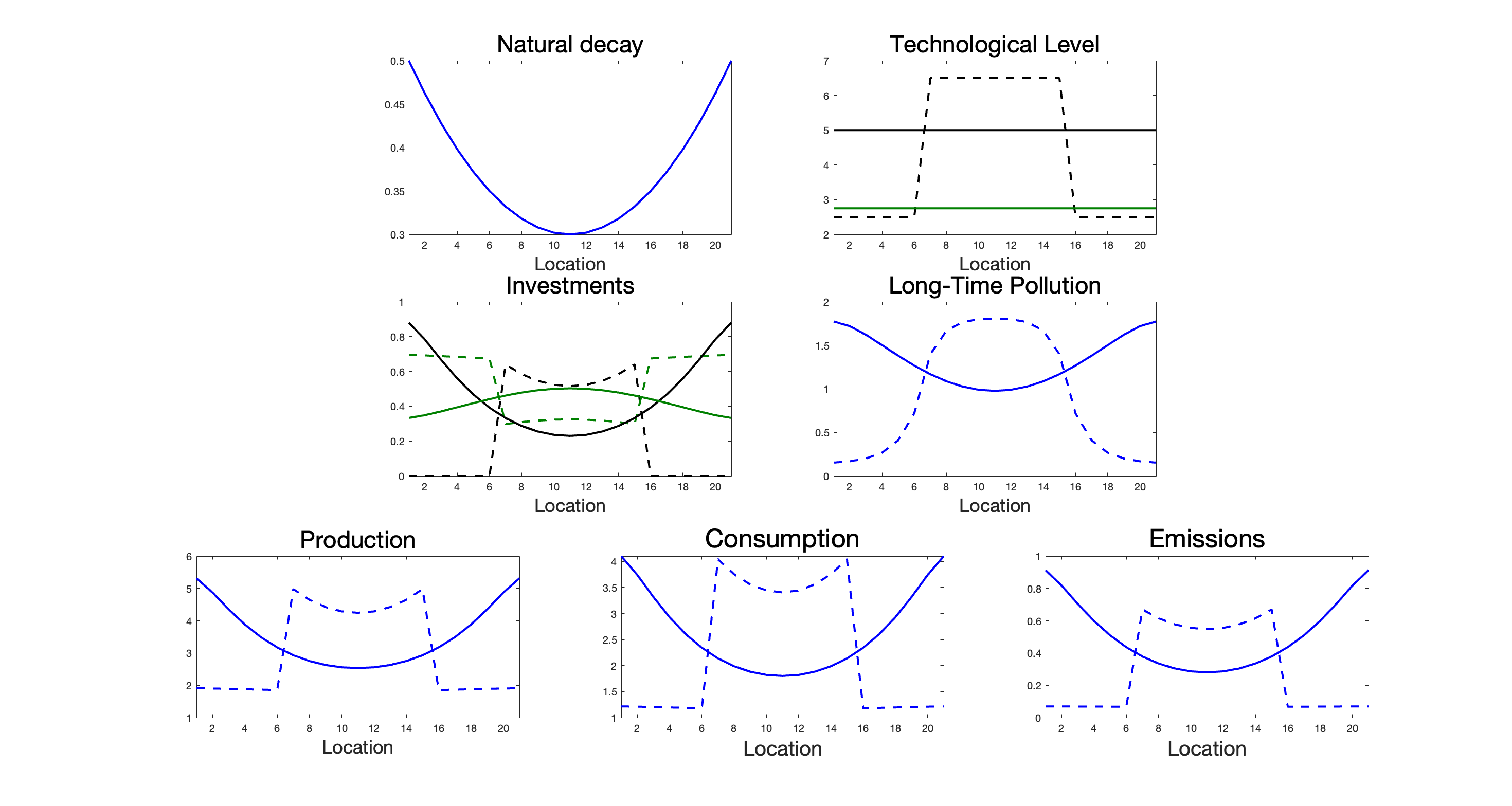}
    \caption{Comparison of optimal investment choices, long-term pollution levels, production, consumption and emissions across nodes.  The parameter values specific to this figure are:  %$ L=L^1, n=21, \rho = 0.03, \omega=1, \varepsilon=0.1, \lambda=1, \theta=0.2, \gamma=4,$ 
     $a_i^R=2.75$, $\forall \; i \in \cV$. The parameter $\delta$ varies across nodes, namely $\delta_i=0.3$ at the core and $\delta_i=0.5$ at the periphery. The straight lines represent the case in which $a^I_i = 5$, $\forall \; i \in \cV$, and the dashed lines the case in which $a^I_i = 6.5$ on the nodes from $7$ to $15$ and $a^I_i = 2.5$ on the other nodes.}
    \label{fig:5}
\end{figure}
However, in the simulations presented, we observe that the model parameters do not have the same impact on the optimal strategies. In particular, model parameters such as $\delta_i, a_i^I$, and $a_i^R$ fundamentally determine the shape of the optimal strategies, whereas the parameter $\alpha_i$ has a second-level impact on the shape of the optimal strategies. This is a consequence of the fact that pollution appears linearly in the objective function.

\section{Conclusions and directions of further research}
In this paper, we introduce a spatio-temporal model where time is treated as a continuous variable while space is represented discretely through a network structure. Within this framework, we analyse both green and brown investments under the assumption of transboundary pollution.\\
By making certain simplifications regarding the dependence of pollution on the objective function, specifically, assuming a linear relationship between pollution levels and the social planner decisions, we derive explicit analytical controls and pollution trajectories. Beyond these mathematical insights, we explore optimal policy implications in the presence of both geographic and economic heterogeneity. In particular, we examine how optimal investment strategies evolve when a new renewable investment option is introduced. We also investigate disparities in pollution decay rates, variations in the diffusion matrix (to account for different pollution flow dynamics), and differences in input productivity, along with their combined effects.\\
However, several aspects remain unexplored in the current model and deserve further study.\\
Observe that, since the regulator is internalizing pollution flow and also that production is generating pollution, we can say that the free riding is prevented. It could be interesting to study the strategic perspective, (see \cite{DeFrutos1, DeFrutos2,BFFGGEB} for some investigations in this directions and also \cite{la2021transboundary,augeraud2025optimal} for the global and local dimension of policies).\\
To enhance realism, a future direction of the work could be to add the possibility that investment at any given moment influences future production. This would mean adding another state variable, the capital, whose state equation describes its accumulation. This would add complexity, in particular if we allow the regulator to manage production exchanges between network nodes. This is part of our current projects.
\\
Another promising direction involves establishing a rigorous mathematical connection between continuous-space models and network-based models (see \cite{DeFrutos2} for work in similar direction). Such theoretical results could serve as a foundation for approximating continuous spatial dynamics using discrete schemes or vice versa.

\appendix

\section{Proofs}

\subsection{Proof of Proposition \ref{welldef}}    \label{appendix:prop_welldef}
\begin{proof}
    The term $ \frac{((A^I(t)-1)I(t)+(A^R(t)-1))^{1-\gamma}}{1-\gamma}$ in \eqref{19} is always either positive (if $\gamma \in (0,1)$) or negative (if $\gamma >1$). Since by the definition of $\cA(p)$, the map $t \to e^{-\rho t}\|f(R_t)\|$ is integrable, it suffices to show that $\int_0^t e^{-\rho t}\langle \omega,P(t) \rangle dt$ is well defined and finite. We have 
    \begin{equation*}
        \int_0^\infty e^{-\rho t} \langle \omega,P(t) \rangle dt=\int_0^\infty e^{-\rho t} \langle \omega, \Phi(t,0) p + \int_0^t\Phi(t,s)N(s)ds  \rangle dt.
    \end{equation*}
    Now since $\omega, p \in \bR^n$ and $\|\Phi(t,s)\| \le 1$ , the first integral term is finite. Moreover for $T>0$ we get, by Fubini-Tonelli's theorem: 
        \begin{equation*}
        \begin{aligned}
             \int_0^T\left( \int_0^t e^{-\rho t}\langle \omega, \Phi(t,s) N(s) \rangle ds \right ) dt &= \int_0^T\left( \int_0^t e^{-\rho s}\langle \omega, e^{-\rho(t-s)}\Phi(t,s)N(s) \rangle ds \right ) dt\\
             &= \int_0^T e^{-\rho s} \langle \omega, N(s)\int_s^T  e^{-\rho(t-s)}\Phi(t,s)dt \rangle ds.
        \end{aligned}
    \end{equation*}
   By sending $T$ to $+\infty$ and recalling that $\| \Phi(t,s) \| \le 1$, the integral $\int_s^\infty e^{- \rho (t-s)} \Phi(t,s) dt$ is bounded. Since $\int_0^\infty e^{-\rho s}\|N(s)\|ds < \infty$, the entire expression is finite.
    \end{proof}
\textcolor{black}{ \subsection{Proof of Proposition \ref{prop:alpha_bound}} \label{appendix:alpha_bound}
\begin{proof}
Define the column sum of the state transition matrix for fixed $j$, $c_j(t) = \sum\nolimits_{i=1}^n \Phi_{i,j}(t,s),$ where  \(\Phi_{i,j}(t,s)\) denotes the \((i,j)\)-th entry of  \(\Phi(t,s)\). By differentiating with respect to $t$:
\begin{equation*} \begin{aligned}
    \frac{d}{dt}c_j(t) &= \sum\nolimits_{i=1}^n \frac{d}{dt}\Phi_{i,j}(t,s) = \sum\nolimits_{i=1}^n \sum\nolimits_{l=1}^n \cL_{i,l}\Phi_{l,j}(t,s) =   \sum\nolimits_{l=1}^n \sum\nolimits_{i=1}^n \cL_{i,l}\Phi_{l,j}(t,s) \\
    &= \sum\nolimits_{l=1}^n \Phi_{l,j}(t,s) \sum\nolimits_{i=1}^n (L_{i,l}(t) - \delta_l(t)\mathbbm{1}_{\{l = i\}}) =  - \sum\nolimits_{i=1}^n  \Phi_{l,j}(t,s) \delta_l(t),
\end{aligned}\end{equation*}
where we used that the sum over columns of $L$ is zero. By calling  $\delta_{\min} = \min_{j \in \cV} \inf_{ s \ge 0} \delta_j(s)$ and $ \delta_{\max} = \max_{j \in \cV} \sup_{ s \ge 0} \delta_j(s)$ it follows 
that $$ - \delta_{\max} c_j(t) \le \frac{d}{dt} c_j(t) \le - \delta_{\min}c_j(t)$$with initial condition
$c_j(s) = \sum\nolimits_{i=1}^n \Phi_{i,j}(s,s) = \sum\nolimits_{i=1}^n I_{i,j} = 1.
$ 
By the comparison principle for ODEs, this implies that
\[
e^{- \delta_{\max}(t-s)} \le c_j(t) \le e^{-\delta_{\min}(t-s)}.
\]
Now consider the $j-$th element of the $\alpha$ vector: 
\begin{equation*}
    \begin{aligned}
    \alpha_j(s) &= \int_s^\infty e^{- \rho (t-s)} [ \Phi^\intercal(t,s) \omega ]_j  dt  =  \int_s^\infty e^{- \rho (t-s)} \sum\nolimits_{i=1}^n \Phi_{i,j}(t,s)  \omega_i dt. \\ 
        \le& \max_{i \in \cV} \omega_i\int_s^\infty e^{- \rho (t-s)} \sum\nolimits_{i=1}^n \Phi_{i,j}(t,s)    = \max_{i \in \cV} \omega_i  \int_s^\infty e^{- \rho (t-s)} c_j(t)  dt \\\le&  \max_{i \in \cV} \omega_i  \int_s^\infty e^{- (\rho+ \delta_{\min}) (t-s)}   dt = \frac{\max_{i \in \cV} \omega_i }{\rho + \delta_{\min}}.
    \end{aligned}
\end{equation*}
The other bound comes symmetrically. 
\end{proof}
}

\subsection{Proof of Proposition \ref{prop:J_rewritten}} \label{appendix:J_rewritten}
\begin{proof}
    Using \eqref{18}, we can rewrite the second term of the functional \eqref{19} in a more convenient way
    \begin{equation}
       \begin{aligned}
            \int_0^\infty e^{-\rho t} \langle w,P(t) \rangle dt&=\int_0^\infty e^{-\rho t} \langle \omega, \Phi(t,0)p + \int_0^t \Phi(t,s)N(s)ds  \rangle dt\\&=\langle \omega, \int_0^\infty e^{-\rho t}\Phi(t,0)pdt\rangle + \int_0^\infty e^{-\rho t}\langle \omega, \int_0^t\Phi(t,s)N(s)ds\rangle dt.
       \end{aligned}
       \label{24}
    \end{equation}
    Note that the first term of the right-hand side is the only one which depends on the initial datum $p$. By \eqref{16}, the first term can be rewritten as 
    \begin{equation*}
        \langle \omega, \int_0^\infty e^{-\rho t}\Phi(t,0)pdt\rangle = \langle  \int_0^\infty e^{-\rho t}\Phi^\intercal(t,0)\omega dt ,p \rangle = \langle \alpha(0), p \rangle.
    \end{equation*}
    The second term in \eqref{24} can be rewritten by exchanging the integrals as: 
    \begin{equation*}
        \begin{aligned}
            \int_0^\infty e^{-\rho t}\langle \omega, \int_0^t \Phi(t,s)N(s)ds\rangle dt &= \int_0^\infty \left ( \int_0^t e^{-\rho t}\langle \omega, \Phi(t,s)N(s)\rangle ds \right ) dt = \\ = \int_0^\infty \left ( \int_0^t e^{-\rho s}\langle \omega, e^{-\rho(t-s)}\Phi(t,s)N(s)\rangle ds \right ) dt &= \int_0^\infty e^{-\rho s}\left < \omega, \int_s^\infty e^{-\rho(t-s)}\Phi(t,s)N(s)dt \right > ds \\ =  \int_0^\infty e^{-\rho s}\left < \int_s^\infty e^{-\rho(t-s)}\Phi^\intercal(t,s)\omega dt, N(s)\right > ds &=  \int_0^\infty e^{-\rho s}\left < \alpha(s), N(s)\right > ds
        \end{aligned}
    \end{equation*}
\end{proof}
\subsection{Proof of Theorem \ref{cor:J_rewritten}} \label{appendix:cor_J_rewritten}
\begin{proof}
    This is a straightforward consequence of Proposition \ref{prop:J_rewritten}. Indeed, according to the reformulation presented in \eqref{J} the problem is reduced to a static one because the integral in \eqref{J} can be optimized pointwise, fixed time $t\in\mathbbm{R}$ and fixed $i\in \cV$. The objective function can be rewritten as
    \begin{equation}\label{J_node}
    J(p, (I,R)) = -\langle \alpha(0),p\rangle +\sum\nolimits_{i=1}^n\int_0^{+\infty} e^{-\rho t} F_i(I_i(t),R_i(t)) dt, 
    \end{equation}
    where $F_i$ is defined in \eqref{eq:F}.
    If $(I^*_i(t),R^*_i(t))$ is a maximum of the function $F_i$ that is integrated in time,  then $(I^*(t),R^*(t))$ is a maximum for the control problem without any constraint on the state variable and the control. If moreover $(I^*(t),R^*(t))$ belong to $\cA(p)$, then it is also optimal for {\bf(P)}.
\end{proof}

\subsection{Proof of Theorem \ref{teo:explicit_1}}\label{appendix:teo:explicit_1}

\begin{proof}
    As suggested by Theorem \ref{cor:J_rewritten}, we look for a control that is admissible and that is optimal for the function $F_i$, defined in \eqref{eq:F}. First, we observe that $F_i$ is coercive and concave on $\bR^2_+$. \\
   Indeed, if $\gamma\in(0,1)$, the utility term grows sub-linearly  and we can estimate \[F_i(I_i, R_i )\leq c_1\frac{(I_i+R_i)^{1-\gamma}}{1-\gamma}
    - \alpha_i(t) (I_i -\varepsilon_i(t) R_i ) -f_i(R_i),\]
    while if $\gamma \in(1,+\infty)$ 
    \[F_i(I_i, R_i )\leq
    - \alpha_i(t) I_i -\varepsilon_i(t) R_i-f_i(R_i).\]
    Since $F_i$ is coercive on the set $\bR^2_+$, $F_i$ admits a global maximum on $\bR^2_+$.
    It is standard to prove that the function $F_i$ is concave on $\bR^2_+$ and it is strictly concave in the inner of $\bR^2_+$. To prove the latter, one can calculate the Hessian matrix 
   \[ H_f=\left[\begin{array}{ccc}
-A\left(a_{i}^I(t)-1\right)^2 & -A\left(a_{i}^I(t)-1\right)\left(a_{i}^R(t)-1\right) \\
-A\left(a_{i}^I(t)-1\right)\left(a_{i}^R(t)-1\right) & -A\left(a_{i}^R(t)-1\right)^2-f''(R_i(t))
\end{array}\right]\]
where $A=\gamma((a_i^I(t)-1)I_i+(a_i^R(t)-1)R_i)^{-\gamma-1}$, and by the signs of all principal minors can deduce the strict concavity of the function in the inner of $\bR^2_+$.\\
We look for the maximum between points that are not differentiable and the points satisfying the Karush-Kuhn-Tucker condition. \\
First, we observe that the points where consumption is null cannot be maximum points.\\
Assume $\gamma<1$, and choose a point with consumption null, namely $(0 ,0)$. Then, we observe that by increasing $\delta$ in the direction of $I_i$, the function $F_i$ increases as well. Indeed
\[F_{i,t}\left(0,0\right)=0, \quad
F_{i,t}\left(\delta,0\right)
=\frac{\left((a_i^I(t)-1)\delta\right)^{1-\gamma}}{1-\gamma}-\alpha\delta
>F_{i,t}\left(0,0\right)
\]
for $\delta$ small enough. A similar argument holds for $\gamma>1$.\\
If condition \eqref{cond_1} holds, then the critical points of the Lagrangian function belong to the inner $\bR^2_+$. Indeed, the system 
    \begin{equation*}
    \begin{cases}
    \left((a^I_i(t)-1)I_i+(a^R_i(t)-1)R_i\right)^{-\gamma}(a^I_i(t)-1)-\alpha_i(t)=0\\
     \left((a^I_i(t)-1)I_i+(a^R_i(t)-1)R_i\right)^{-\gamma}\cdot(a^R_i(t)-1)- \varepsilon_i(t)\alpha_i(t)-f'(R_i)=0\\
    \end{cases}
    \end{equation*}
    admits solution $(I_i,R_i)=\left(\left[\left(\frac{\alpha_i(t)}{a_i^I(t)-1}\right)^{-1/\gamma}-(a_{i}^R(t)-1)R^*\right]\frac{1}{(a_{i}^I(t)-1)},(f'_i)^{-1}\left(\alpha_i(t)\left(\frac{a_{i}^R(t)-1}{a_i^I(t)-1}-\varepsilon_i(t)\right)\right)\right)$
where condition \eqref{cond_1} guarantees the positivity of $R_i$, while condition \eqref{cond_2} guarantees the positivity of $I_i$. \\
\textcolor{black}{Now we need to prove that $(R_i,I_i)$ belongs to the set of admissible controls. Specifically, the optimal net emission term can be explicitly written as
\[
N_i(t) = \alpha_i(t)^{-\frac{1}{\gamma}} (a_i^I(t) - 1)^{\frac{1-\gamma}{\gamma}} + \left[\varepsilon_i(t) - \frac{a_i^R(t) - 1}{a_i^I(t) - 1}\right] (f_i')^{-1}\left( \alpha_i(t) \left(\frac{a_i^R(t) - 1}{a_i^I(t) - 1} - \varepsilon_i(t) \right) \right).
\]
Define
$
x(t) := \alpha_i(t) \left( \frac{a_i^R(t) - 1}{a_i^I(t) - 1} - \varepsilon_i(t) \right).
$
Condition \eqref{cond_1} gives \(x(t) > 0\) and by the inverse function rule, since \(f_i'' \ge \epsilon > 0\), \((f_i')^{-1}\) is Lipschitz continuous. Together with assumption \ref{assumption:renewable} $(ii)$, this allows us to bound 
$\|N(t)\| \le C e^{g t},
$ 
and by \ref{assumption:renewable} $(iii)$ we conclude.\\
}
If \eqref{cond_3} holds, the inner of $\bR^2_+$ does not contain any critical points. In particular, the solution of the lagrangian system    
    \begin{equation*}
    \begin{cases}
    \left((a^I_i(t)-1)I_i+(a^R_i(t)-1)R_i\right)^{-\gamma}(a^I_i(t)-1)-\alpha_i(t)=0\\
     \left((a^I_i(t)-1)I_i+(a^R_i(t)-1)R_i\right)^{-\gamma}\cdot(a^R_i(t)-1)- \varepsilon_i(t)\alpha_i(t)-f'(R_i)=-\mu_{i}^R, \quad \mu_{i}^R\geq 0\\
    \end{cases}
    \end{equation*}
    is $(I_i,R_i)=\left(\left[\left(\frac{\alpha_i(t)}{a_i^I(t)-1}\right)^{-1/\gamma}\right]\frac{1}{(a_{i}^I(t)-1)},0\right)$ with $\mu_i^R=\alpha_i(t)\left(\varepsilon_i(t)-\frac{(a_i^R(t)-1)}{(a_i^I(t)-1)}\right)+f_i'(0)$.
    Condition \eqref{cond_3} guarantees that the lagrangian multiplication $\mu^R_i$ is positive \textcolor{black}{and Assumption \ref{assumption:renewable} ensures that $(I_i,R_i) \in \cA(p)$.} 
    To conclude, when condition \eqref{cond_1} is verified and condition \eqref{cond_2} is not, namely
    \begin{equation*}
        \left(\frac{\alpha_i(t)}{a_i^I(t)-1}\right)^{-1/\gamma}> (a_{i}^R(t)-1)(f'_i)^{-1}\left(\alpha_i(t)\left(\frac{a_i^R(t)-1}{a_i^I(t)-1}-\varepsilon_i(t)\right)\right)\\
    \end{equation*}
    the inner of $\bR^2_+$ does not contain any critical points and the global maximum belongs to the boundary, and since condition \eqref{cond_1} is assumed to hold, the maximum is solution of the lagrangian system
     \begin{equation*}
    \begin{cases}
    \left((a^R_i(t)-1)R_i\right)^{-\gamma}(a^I_i(t)-1)-\alpha_i(t)=-\mu_i^I\\
     \left((a^R_i(t)-1)R_i\right)^{-\gamma}\cdot(a^R_i(t)-1)- \varepsilon_i(t)\alpha_i(t)-f'(R_i)=0\\
    \end{cases}
    \end{equation*}
    However, in this case, the optimal controls are not explicitly derived. Indeed, they are solutions of the non-linear system
    \begin{equation*}
    \begin{cases}

R_i\,\,s.t.\,\, \left((a^R_i(t)-1)R_i\right)^{-\gamma}\cdot(a^R_i(t)-1)- \varepsilon_i(t)\alpha_i(t)=f'(R_i).\\
I_i=0.
    \end{cases}
    \end{equation*}
    Notice that, since the function
    $\Phi(x)= \left((a^R_i(t)-1)x\right)^{-\gamma}\cdot(a^R_i(t)-1)- \varepsilon_i(t)\alpha_i(t)-f'(x),\\
$
is monotonous strictly decreasing and $\lim_{x\to 0^+}\Phi(x)=\infty$, $\lim_{x\to+\infty} \Phi(x)=-\infty$, we conclude that the $\Phi(x)$ as only one zero and therefore $R^*_i$ is unique. \textcolor{black}{The admissibility of the optimal point follows from Assumption~\ref{assumption:renewable} (ii--iii), together with the facts that $\alpha_i(t)$ is uniformly bounded,  and $f_i' \ge \epsilon > 0$.}
\end{proof}

\subsection{Proof of Theorem \ref{teo:explicit_2}}
\label{appendix:teo:explicit_2}
\begin{proof}
    By following the same argument proposed in Theorem \ref{teo:explicit_1}, it is possible to prove that by choosing $f_i$ as linear $F_i$ is coercive on the set $\bR^2_+$, and therefore $F_i$ admits a global maximum on $\bR^2_+$.  We look for the maximum between points that are not differentiable and the points satisfying the Karush-Kuhn-Tucker condition. By adapting the argument presented in Theorem \ref{teo:explicit_1}, one can exclude points where the function $F_i$ is not differentiable. In particular, one can observe that points where consumption is null cannot be maximum points.\\
    Therefore, the global maximum is a solution of the KKT system.\\
    If condition \eqref{cond_3_linear} holds, the critical points of the Lagragian belongs to the inner of $\bR^2_+$. In particular, the system 
    \begin{equation*}
    \begin{cases}
    \left((a^I_i(t)-1)I_i+(a^R_i(t)-I_i)R_i\right)^{-\gamma}(a^I_i(t)-1)-\alpha_i(t)=0\\
     \left((a^I_i(t)-1)I_i+(a^R_i(t)-I_i)R_i\right)^{-\gamma}\cdot(a^R_i(t)-1)- \varepsilon_i(t)\alpha_i(t)-\lambda_i=0
    \end{cases}
    \end{equation*}
 admits as solutions the points of the line
\[
    (a_i^I(t)-1)I+(a_i^R(t)-1)R_i=\left(\alpha_i(t)\cdot(a_i^I(t)-1)^{-1}\right)^{-\frac{1}{\gamma}}.\\
\]
If condition \eqref{cond_1} holds, the lagrangian system
 \begin{equation*}
    \begin{cases}
    \left((a^I_i(t)-1)I_i+(a^R_i(t)-I_i)R_i\right)^{-\gamma}(a^I_i(t)-1)-\alpha_i(t)=-\mu_i^I, \quad \mu_i^I\geq 0\\
     \left((a^I_i(t)-1)I_i+(a^R_i(t)-I_i)R_i\right)^{-\gamma}\cdot(a^R_i(t)-1)- \varepsilon_i(t)\alpha_i(t)-\lambda_i=0
    \end{cases}
    \end{equation*}
 admits as solutions the point $(I_i,R_i)=\left(0,\frac{1}{(a^R_i(t)-1)}\left(\frac{\lambda_i+\alpha_i(t)\varepsilon_i(t)}{(a^R_i(t)-1)}\right)^{-\frac1\gamma}.\right)$
To conclude, if condition \eqref{cond_3} holds, the lagrangian system
 \begin{equation*}
    \begin{cases}
    \left((a^I_i(t)-1)I_i+(a^R_i(t)-I_i)R_i\right)^{-\gamma}(a^I_i(t)-1)-\alpha_i(t)=0\\
     \left((a^I_i(t)-1)I_i+(a^R_i(t)-I_i)R_i\right)^{-\gamma}\cdot(a^R_i(t)-1)- \varepsilon_i(t)\alpha_i(t)-\lambda_i=-\mu_i^R, \quad \mu_i^R\geq 0
    \end{cases}
    \end{equation*}
 admits as solutions the point $(I_i,R_i)=\left(\frac{1}{(a^I_i(t)-1)}\left(\frac{\alpha_i(t)}{(a^I_i(t)-1)}\right)^{-\frac1\gamma},0\right)$
\end{proof}
%\newpage
%    \textcolor{orange}{
%\begin{itemize}
%\item Se consideriamo il termine $\langle \omega, P(t)\rangle$ come una sorta di carbon tax, potremmo risolvere poi un problema successivo, tipo:
%   \[\min_{\omega}P^*(T)\]
%\end{itemize}
%}
\subsection{Proof of Theorem \ref{teo:P_infty}} \label{appendix:P_infty}
\begin{proof}
    Notice that in this case, the expressions of the optimal controls are time-independent, too. 
%\begin{equation*}    \begin{aligned}            &B^*(t)\equiv B^* =(\theta \varphi(A-\mathbbm{1}))^{\frac{1}{1-\theta}},\\        &I^*(t) \equiv I^*=\alpha^{-\frac{1}{\gamma}}(A-\mathbbm{1}))^{\frac{1-\gamma}{\gamma}}+(\theta\varphi^{\frac{1}{1-\theta}}(A-\mathbbm{1}))^{\frac{1}{1-\theta}}.  \end{aligned}\end{equation*}
    Let $\theta_0$ be the spectral bound of $\cL$. Since $\delta > 0$, the operator $\cL$ is strictly dissipative, hence $\theta_0 < 0.$ Let us write 
    $\cL = \cL_0 -\theta_0$ where $\cL_0 \coloneqq \cL + \theta_0,
$   and note that $\cL_0$ is dissipate by definition, hence $e^{s\cL_0}$ is a contraction. Then, we can rewrite:
    \begin{equation*}
        \begin{aligned}
            P^*(t) \coloneqq & e^{t\cL_0}e^{-\theta_0 t}p + \int_0^t e^{-\theta_0(t-s)}e^{(t-s)\cL_0}N^*ds=e^{t\cL_0}e^{-\theta_0 t}p + \int_0^t e^{-\theta_0s}e^{s\cL_0}N^*ds, \quad t\ge0,
        \end{aligned}
    \end{equation*}
    and take the limit above when $t \to \infty$. Since $e^{s\cL_0}$ is a contraction, the first one on the right-hand side converges to $0$, whereas the second one converges to 
    $    P^*_\infty \coloneqq \int_0^\infty e^{-\theta_0s}e^{s\cL_0}N^*ds, \quad t\ge0.
    $
    And we can conclude by expressing the limit $P^*_\infty$ as $P^*_\infty = (\theta_0 - \cL_0)^{-1}N^*$. I.e. $P^*_\infty$ is the solution to $(\theta_0 - \cL_0)P=N^*$ or, equivalently, to $\cL P + N^* = 0$.
    \end{proof}

\subsection{Proof of Theorem \ref{theo:3.3}} \label{appendix:3.3}
\begin{proof}
As suggested by Theorem \ref{cor:J_rewritten}, we look for a control that is admissible and that is optimal for the function $F$, defined in \eqref{eq:F}, where the control $R$ is zero, namely
\begin{align*}
     F(I)= \left<  \frac{((A^I(t)-1)I)^{1-\gamma}}{1-\gamma}, \mathbf{1} \right >
    - \langle \alpha(t), I  \rangle
=\sum\nolimits_{i=1}^n\frac{((a^I_i(t)-1)I_i)^{1-\gamma}}{1-\gamma}
    - \alpha_i(t) I_i =\sum\nolimits_{i=1}^nF_i(I_i).
     \end{align*}

    First, we need to check that $I^* \in \cA(p)$. We have 
    \begin{equation*}
        (A^I(t)-\mathbbm{1}))I^*(t)  = \left(\frac{A^I(t)-\mathbbm{1}}{\alpha(t)}\right)^{\frac{1}{\gamma}} \ge 0, \quad \forall t \in \bR_+.
    \end{equation*}
    Moreover, considering $N^*(t)$ as in \eqref{eq:N_optimal_noR} and Assumption \ref{ass2}, we get the existence of some constant $C_0 > 0$ such that $
        0 \le N_i^*(t) \le C_0e^{gt}, $ $\forall t \in \bR_+, \; \forall i \in \cV.$
    We conclude that $I^*\in \cA(p)$.\\
    Concerning optimality, as stressed in Theorem \ref{cor:J_rewritten}, the integrals in \eqref{J} can be optimized pointwisely, indeed fix $t \in \bR_+$, $i \in \cV$. \\ 
    By strict concavity of $F_i$ with respect to $I_i(t)$, the unique maximum point can be found just by first-order optimality conditions:
$
            ((a^I_i(t)-1)I_i(t))^{-\gamma}(a^I_i(t)-1)-\alpha_i(t)=0. 
            $
            
    The claim on the optimal control then follows by solving the above system, and all the remaining claims immediately follow from straightforward computations. 
\end{proof}
\nocite{*}
\bibliographystyle{apalike}
\bibliography{biblio_main}

@article{Bull,
  title={Renewable energy today and tomorrow},
  author={Bull, Stanley R},
  journal={Proceedings of the IEEE},
  volume={89},
  number={8},
  pages={1216--1226},
  year={2001},
  publisher={IEEE}
}

@article{bartaloni2019infinite,
  title={Infinite horizon optimal control problems with non-compact control space. Existence results and dynamic programming},
  author={Bartaloni, Francesco},
  year={2019}
}

@article{BFFGEJOR,
  title={From firm to global-level pollution control: The case of transboundary pollution},
  author={Boucekkine, Raouf and Fabbri, Giorgio and Federico, Salvatore and Gozzi, Fausto},
  journal={European journal of operational research},
  volume={290},
  number={1},
  pages={331--345},
  year={2021},
  publisher={Elsevier}
}

@article{LiYong,
  title={Control Problems in Infinite Dimensions},
  author={Li, Xunjing and Yong, Jiongmin },
  journal={Optimal Control Theory for Infinite Dimensional Systems},
  pages={1--23},
  year={1995},
  publisher={Springer}
}

@article{BFFGJME,
  title={Growth and agglomeration in the heterogeneous space: a generalized AK approach},
  author={Boucekkine, Raouf and Fabbri, Giorgio and Federico, Salvatore and Gozzi, Fausto},
  journal={Journal of Economic Geography},
  volume={19},
  number={6},
  pages={1287--1318},
  year={2019},
  publisher={Oxford University Press}
}

@article{CGLPXY,
title = {An optimal control problem with state constraints in a spatio-temporal economic growth model on networks},
journal = {Journal of Mathematical Economics},
volume = {113},
pages = {102991},
year = {2024},
issn = {0304-4068},
doi = {https://doi.org/10.1016/j.jmateco.2024.102991},
author = {Alessandro Calvia and Fausto Gozzi and Marta Leocata and Georgios I. Papayiannis and Anastasios Xepapadeas and Athanasios N. Yannacopoulos},
}

@article{BFFGGEB,
  title={A dynamic theory of spatial externalities},
  author={Boucekkine, Raouf and Fabbri, Giorgio and Federico, Salvatore and Gozzi, Fausto},
  journal={Games and Economic Behavior},
  volume={132},
  pages={133--165},
  year={2022},
  publisher={Elsevier}
}

@book{engel2000one,
  title={One-parameter semigroups for linear evolution equations},
  author={Engel, Klaus-Jochen and Nagel, Rainer and Brendle, Simon},
  volume={194},
  year={2000},
  publisher={Springer}
}

@article{freni2008optimal,
  title={Optimal strategies in linear multisector models: Value function and optimality conditions},
  author={Freni, Giuseppe and Gozzi, Fausto and Pignotti, Cristina},
  journal={Journal of Mathematical Economics},
  volume={44},
  number={1},
  pages={55--86},
  year={2008},
  publisher={Elsevier}
}

@article{freni2006existence,
  title={Existence of optimal strategies in linear multisector models},
  author={Freni, Giuseppe and Gozzi, Fausto and Salvadori, Neri},
  journal={Economic Theory},
  volume={29},
  pages={25--48},
  year={2006},
  publisher={Springer}
}

@book{Troltsch,
  title={Optimal control of partial differential equations: theory, methods and applications},
  author={Tr{\"o}ltzsch, Fredi},
  volume={112},
  year={2024},
  publisher={American Mathematical Society}
}

@article{DeFrutos1,
  title={Spatial effects and strategic behavior in a multiregional transboundary pollution dynamic game},
  author={de Frutos, Javier and Mart{\'\i}n-Herr{\'a}n, Guiomar},
  journal={Journal of Environmental Economics and Management},
  volume={97},
  pages={182--207},
  year={2019},
  publisher={Elsevier}
}

@article{DeFrutos2,
  title={Spatial vs. non-spatial transboundary pollution control in a class of cooperative and non-cooperative dynamic games},
  author={de Frutos, Javier and Mart{\'\i}n-Herr{\'a}n, Guiomar},
  journal={European Journal of Operational Research},
  volume={276},
  number={1},
  pages={379--394},
  year={2019},
  publisher={Elsevier}
}

@article{DeFrutos3,
  title={Equilibrium strategies in a multiregional transboundary pollution differential game with spatially distributed controls},
  author={de Frutos, Javier and L{\'o}pez-P{\'e}rez, Paula M and Mart{\'\i}n-Herr{\'a}n, Guiomar},
  journal={Automatica},
  volume={125},
  pages={109411},
  year={2021},
  publisher={Elsevier}
}

@article{XUEWANGJCP24,
  title={The impact of pollution transmission networks in a transboundary pollution game},
  author={Xue, Linzhao and Wang, Xianjia},
  journal={Journal of Cleaner Production},
  volume={451},
  pages={142010},
  year={2024},
  publisher={Elsevier}
}

@article{de2020non,
  title={Non-linear incentive equilibrium strategies for a transboundary pollution differential game},
  author={de Frutos, Javier and Mart{\'\i}n-Herr{\'a}n, Guiomar},
  journal={Games in Management Science: Essays in Honor of Georges Zaccour},
  pages={187--204},
  year={2020},
  publisher={Springer}
}

@article{de2022investment,
  title={Investment in cleaner technologies in a transboundary pollution dynamic game: A numerical investigation},
  author={de Frutos, Javier and Gat{\'o}n, V{\'\i}ctor and L{\'o}pez-P{\'e}rez, Paula M and Mart{\'\i}n-Herr{\'a}n, Guiomar},
  journal={Dynamic Games and Applications},
  volume={12},
  number={3},
  pages={813--843},
  year={2022},
  publisher={Springer}
}

@article{jorgensen2010dynamic,
  title={Dynamic games in the economics and management of pollution},
  author={J{\o}rgensen, Steffen and Mart{\'\i}n-Herr{\'a}n, Guiomar and Zaccour, Georges},
  journal={Environmental Modeling \& Assessment},
  volume={15},
  pages={433--467},
  year={2010},
  publisher={Springer}
}

@book{farina2011positive,
  title={Positive linear systems: theory and applications},
  author={Farina, Lorenzo and Rinaldi, Sergio},
  year={2011},
  publisher={John Wiley \& Sons}
}

@book{rugh1996linear,
  title={Linear system theory},
  author={Rugh, Wilson J},
  year={1996},
  publisher={Prentice-Hall, Inc.}
}

@article{kaczorek2015class,
  title={A class of positive and stable time-varying electrical circuits},
  author={Kaczorek, Tadeusz},
  journal={Electrical Review},
  volume={91},
  number={5},
  pages={121--124},
  year={2015}
}

@article{bawa,
  title={On optimal pollution control policies},
  author={Bawa, Vijay S},
  journal={Management Science},
  volume={21},
  number={12},
  pages={1397--1404},
  year={1975},
  publisher={INFORMS}
}

@article{kerl2015new,
  title={New approach for optimal electricity planning and dispatching with hourly time-scale air quality and health considerations},
  author={Kerl, Paul Y and Zhang, Wenxian and Moreno-Cruz, Juan B and Nenes, Athanasios and Realff, Matthew J and Russell, Armistead G and Sokol, Joel and Thomas, Valerie M},
  journal={Proceedings of the National Academy of Sciences},
  volume={112},
  number={35},
  pages={10884--10889},
  year={2015},
  publisher={National Academy of Sciences}
}

@article{xu2009local,
  title={Local air pollutant emission reduction and ancillary carbon benefits of SO2 control policies: Application of AIM/CGE model to China},
  author={Xu, Yan and Masui, Toshihiko},
  journal={European Journal of Operational Research},
  volume={198},
  number={1},
  pages={315--325},
  year={2009},
  publisher={Elsevier}
}

@article{zhao2013harmonizing,
  title={Harmonizing model with transfer tax on water pollution across regional boundaries in a China s lake basin},
  author={Zhao, Laijun and Li, Changmin and Huang, Rongbing and Si, Steven and Xue, Jian and Huang, Wei and Hu, Yue},
  journal={European Journal of Operational Research},
  volume={225},
  number={2},
  pages={377--382},
  year={2013},
  publisher={Elsevier}
}

@article{stam1992transboundary,
  title={Transboundary air pollution in Europe: An interactive multicriteria tradeoff analysis},
  author={Stam, Antonie and Kuula, Markku and Cesar, Herman},
  journal={European Journal of Operational Research},
  volume={56},
  number={2},
  pages={263--277},
  year={1992},
  publisher={Elsevier}
}

@article{bertinelli2014carbon,
  title={Carbon capture and storage and transboundary pollution: A differential game approach},
  author={Bertinelli, Luisito and Camacho, Carmen and Zou, Benteng},
  journal={European Journal of Operational Research},
  volume={237},
  number={2},
  pages={721--728},
  year={2014},
  publisher={Elsevier}
}

@article{leibowicz2020urban,
  title={Urban land use and transportation planning for climate change mitigation: A theoretical framework},
  author={Leibowicz, Benjamin D},
  journal={European Journal of Operational Research},
  volume={284},
  number={2},
  pages={604--616},
  year={2020},
  publisher={Elsevier}
}

@article{el2020transboundary,
  title={Transboundary pollution control and environmental absorption efficiency management},
  author={El Ouardighi, Fouad and Kogan, Konstantin and Gnecco, Giorgio and Sanguineti, Marcello},
  journal={Annals of Operations Research},
  volume={287},
  number={2},
  pages={653--681},
  year={2020},
  publisher={Springer}
}

@article{ferrari2019strategic,
  title={On a strategic model of pollution control},
  author={Ferrari, Giorgio and Koch, Torben},
  journal={Annals of Operations Research},
  volume={275},
  number={2},
  pages={297--319},
  year={2019},
  publisher={Springer}
}

@article{ballester2006s,
  title={Who's who in networks. Wanted: The key player},
  author={Ballester, Coralio and Calv{\'o}-Armengol, Antoni and Zenou, Yves},
  journal={Econometrica},
  volume={74},
  number={5},
  pages={1403--1417},
  year={2006},
  publisher={Wiley Online Library}
}

@inproceedings{elliott2013network,
  title={A network approach to public goods},
  author={Elliott, Matthew and Golub, Benjamin},
  booktitle={Proceedings of the fourteenth ACM conference on Electronic commerce},
  pages={377--378},
  year={2013}
}

@article{nerantzis,
  title={Optimal control of water distribution networks without storage},
  author={Nerantzis, Dimitrios and Pecci, Filippo and Stoianov, Ivan},
  journal={European Journal of Operational Research},
  volume={284},
  number={1},
  pages={345--354},
  year={2020},
  publisher={Elsevier}
}

@article{la2021transboundary,
  title={Transboundary pollution externalities: Think globally, act locally?},
  author={La Torre, Davide and Liuzzi, Danilo and Marsiglio, Simone},
  journal={Journal of Mathematical Economics},
  volume={96},
  pages={102511},
  year={2021},
  publisher={Elsevier}
}

@article{augeraud2025optimal,
  title={Optimal policies for environmental assets under spatial heterogeneity and global awareness},
  author={Augeraud-V{\'e}ron, Emmanuelle and Ghilli, Daria and Gozzi, Fausto and Leocata, Marta},
  journal={arXiv preprint arXiv:2510.21397},
  year={2025}
}
\end{document}